\documentclass[twocolumn,letterpaper]{IEEEAerospaceCLS}  

\usepackage[]{graphicx,float,latexsym,amssymb,amsfonts,amsmath,mathtools,amstext,times}
\usepackage{hyperref}
\usepackage{cite}
\newtheorem{theorem}{Theorem}
\newtheorem{proposition}{Proposition}
\newtheorem{lemma}{Lemma}

\usepackage{url}
\usepackage[]{graphicx}    
\newcommand{\ignore}[1]{}  

\begin{document}
\title{Pitch Plane Trajectory Tracking Control for Sounding Rockets via Adaptive Feedback Linearization}

\author{%
Pedro dos Santos\\ 
IDMEC, IST, University of Lisbon\\
Av. Rovisco Pais 1\\
Lisbon, Portugal 1049-001\\
pedrodossantos31@tecnico.ulisboa.pt
\and 
Paulo Oliveira\\
IDMEC, IST, University of Lisbon\\
Av. Rovisco Pais 1\\
Lisbon, Portugal 1049-001\\
paulo.j.oliveira@tecnico.ulisboa.pt
\thanks{\footnotesize 979-8-3503-5597-0/25/$\$31.00$ \copyright2025 IEEE}              
}

\maketitle

\thispagestyle{plain}
\pagestyle{plain}

\maketitle

\thispagestyle{plain}
\pagestyle{plain}

\begin{abstract}
This paper proposes a pitch plane trajectory tacking control solution for suborbital launch vehicles relying on adaptive feedback linearization. Initially, the 2D dynamics and kinematics for a single-engine, thrust-vector-controlled sounding rocket are obtained for control design purposes. Then, an inner-outer control strategy, which simultaneously tackles attitude and position control, is adopted, with the inner-loop comprising the altitude and pitch control and the outer-loop addressing the horizontal (downrange) position control. Feedback linearization is used to cancel out the non-linearities in both the inner and outer dynamics. Making use of Lyapunov stability theory, an adaptation law, which provides online estimates on the inner-loop aerodynamic uncertainty, is jointly designed with the output tracking controller via adaptive backstepping, ensuring global reference tracking in the region where the feedback linearization is well-defined. The zero dynamics of the inner-stabilized system are then exploited to obtain the outer-loop dynamics and derive a Linear Quadratic Regulator (LQR) with integral action, which can stabilize them as well as reject external disturbances. In the outermost loop, the estimate on the correspondent aerodynamic uncertainty is indirectly obtained by using the inner loop estimates together with known aerodynamics relations. The resulting inner-outer position control solution is proven to be asymptotically stable in the region of interest. Finally, the control strategy is implemented in a Matlab/Simulink simulation environment composed by the non-linear pitch plane dynamics and kinematics model and the environmental disturbances to assess its performance. Using a single-stage sounding rocket, propelled by a liquid engine, as reference vehicle, different mission scenarios are tested in the simulation environment to verify the adaptability of the proposed control strategy. The system is able to track the requested trajectories while rejecting external wind disturbances. Furthermore, the need to re-tune the control gains in between different mission scenarios is minimal to none.
\end{abstract}

\tableofcontents

\section{Introduction}
In the last decades, suborbital launch vehicles endowed the scientific community with tools to perform a
myriad of research studies \cite{nasa_sr,esa_sr}. Commonly denominated as sounding rockets, they provide long periods of
microgravity conditions, allow to collect in-situ data across all atmospheric layers, and enable rapid Earth
surveillance and monitoring \cite{nasa_sr, esa_sr,noga}. Simultaneously, they can be instrumental as low-cost testing platforms to
augment Technology Readiness Levels (TRL) of different systems and payloads, before its use in high risk, potentially crewed,
orbital/space missions \cite{noga}. More recently, following the successful efforts of private companies, such as SpaceX, Blue Origin, and Rocket Lab, a growing number of both private corporations and national/international agencies, namely at European level, are investing towards a new generation of reusable micro and small launch vehicles, which motivates the development of suborbital platforms for technology demonstration purposes \cite{ariane,simplicio,callisto}. In addition, suborbital transportation and space tourism motivated a market increase which impacts the overall need for cost-effective, dedicated suborbital launchers \cite{spacetourism}.

In face of increasingly demanding mission scenarios, focusing on reusability and reconfigurability, active stabilization and trajectory tracking are now of upmost importance for sounding rockets, which historically relied on passive aerodynamic stabilization. When actively controlled, a sounding rocket generally classifies as an underactuated autonomous vehicle, i.e, the number of control inputs is less than the number of degrees-of-freedom \cite{aguiar}. Particularly, when considering thrust vector control as the only actuation method, the force and torque inputs are coupled, posing a challenging trajectory control task due to the presence of non-holonomic constraints \cite{mahmut}. The classical approach to the trajectory tracking control problem for underactuated launch vehicles relies on the decoupling of the multi-variable dynamics into lower order fully-actuated models, typically by separating the attitude and position dynamics \cite{Tewari}. An external guidance loop controls the position by generating attitude commands and an internal attitude control loop is responsible for tracking those commands. The control design is then performed separately, relying on the time-scale separation between the attitude and position dynamics to neglect the coupled behavior and bypass the non-holonomic constraints. In fact, a clear separation can be found in the literature between both problems: the term ``control'' is usually interchangeable with attitude control, whereas position control is referred to as guidance, and are separately addressed. Traditionally, attitude control is achieved by the application of linear techquines which rely on the local linearization of the model around the nominal trajectory and gain-scheduling \cite{Wie}. These include PID control \cite{pid}, the Linear Quadratic Regulator (LQR) \cite{dossantos}, $H_\infty$ control \cite{simplicio2, sagliano}, and Linear Parameter-Varying (LPV) control \cite{tapia}. As for the guidance loop, recent years have seen an increasing investigation effort on online guidance algorithms able to cope with more demanding mission scenarios \cite{successive_conv}. These algorithms make use of the increase in the available computational power in order to generate in real-time optimal reference trajectories and the required attitude commands to track it.  
 
While trying to exploit the full range of applications for soundings rockets, the classical approach to launch vehicle trajectory control may be rendered ineffective. The lower development costs and overall reduced operational logistics allow for more responsive launch missions \cite{noga}, which may be tailored to requirements that change with the mission objectives, imposing the need to track considerably distinct trajectories in large flight envelopes. Following the classical approach would require an extensive redesign and adaptation of the control system to each mission scenario, since it strongly depends on the trajectory to follow. Moreover, vehicle reconfigurability would pose an additional burden due to the shifting vehicle parameters. On the other hand, the separation between the attitude control and guidance loops hinders the derivation of mathematical proofs that ensure that the final system, i.e., the interconnection of both loops, is stable. It is then clear the need for integrated, adaptive, and global trajectory tracking control solutions which accommodate for reconfigurable vehicle parameters and distinct mission scenarios without requiring extensive redesign. To achieve this, non-linear control techniques, which consider the full dynamics of the system, are instrumental.

Several works can be found in the literature which leverage non-linear control techniques, mostly based on Lyapunov theory \cite{Khalil}, to tackle the launch vehicle attitude control problem. Among others, sliding mode control \cite{sliding_mode}, backstepping \cite{backstepping}, and non-linear dynamic inversion (or feedback linearization) \cite{ndi1, ndi2} are quite common. However, the literature is scarce when it comes to jointly address position and attitude control for launch vehicle trajectory tracking. On the other side, this control approach is standard for other underactuated mechanical systems, namely quadrotors \cite{quadrotor}, autonomous marine vehicles \cite{marine}, missiles \cite{missile}, and guided projectiles \cite{projectile}. In this work, we follow the methodology proposed by one of the authors for quadrotors in \cite{Martins}, where inner-outer loop feedback linearization is applied and the zero dynamics \cite{Isidori} are stabilized to provide horizontal position control. When applied to launch vehicles, feedback linearization is partially used for the attitude control loop and, to overcome model uncertainty, incremental techniques, such as incremental non-linear dynamic inversion (INDI) have gained more relevance \cite{Sieberling2010, Mooij2020}. 

In this paper, a trajectory tracking controller for a liquid engine sounding rocket, with the dynamics restricted to the pitch plane, is proposed, in which position and attitude control are integrated in an inner-outer loop structure. Thrust vector control is used as actuation method, meaning that the control inputs are the thrust magnitude and the thrust vector deflection angle. The inner-loop comprises the altitude and pitch control and the outer-loop addresses the horizontal (downrange) position control. Feedback linearization is used to cancel out the non-linearities in both the inner and outer dynamics, reducing them to two double integrators acting on each of the output tracking variables. Uncertainty is considered when canceling the aerodynamic terms and is estimated in real-time in the inner loop via adaptive backstepping. The zero dynamics of the inner-stabilized system are then exploited to obtain the outer-loop dynamics and derive a Linear Quadratic Regulator (LQR) with integral action, which can stabilize them as well as reject external disturbances.

With respect to the literature, the main contribution of our strategy is the ability to track arbitrary sufficiently smooth reference trajectories, while ensuring global stability in the region where the feedback linearization is well-defined. Moreover, the adaptation scheme tackles the inherent drawback of feedback linearization of requiring accurate knowledge of system parameters for dynamic inversion, which are often not available or have high uncertainty, without the need for incremental techniques. The integrated design of position and attitude control leads to a solution than can easily be re-purposed for different mission scenarios and vehicle configurations with minimal effort.

This paper is structured as follows: in Section \ref{sec:physical_model}, the pitch plane dynamics and kinematics for a sounding rocket with a single gimbaled engine are derived. In Section \ref{sec:control}, the proposed inner-outer trajectory tracking control solution is derived. In Section \ref{sec:implement}, the details of the computational implementation of the control architecture in the simulation model are detailed, while the simulation results are presented in Section \ref{sec:sim_reuslts}. Finally, some conclusions are drawn and reference to future work is made in Section \ref{sec:conc}.

\section{Physical Model}\label{sec:physical_model}

In this section, the pitch plane dynamics and kinematics of a generic launch vehicle with a single gimbaled engine are derived following the Newton-Euler formalism. To obtain the physical model some assumptions are used: the launch vehicle is assumed to be a rigid body; it is assumed to be axially symmetric; and the Earth's curvature and rotation are neglected. All these assumptions are often followed in the literature \cite{Tewari,Wie}, and are considered valid for a first stage design of the pitch plane trajectory tracking control system.

\subsection{Reference Frames} 
\begin{figure}[h]
\centering
\includegraphics{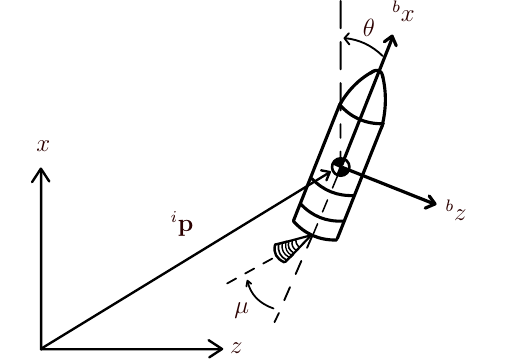}\\
\caption{\textbf{Reference frames.}}
\label{fig:frames}
\end{figure}

As depicted in Fig. \ref{fig:frames}, the model is deduced considering a body frame, \{b\}, attached to the center of mass of the vehicle, and an inertial frame, \{i\}, with its origin located at the launch site. The inertial position of the vehicle, $\prescript{i}{}{\mathbf{p}} = \left[\,x\:\:\:\:z\right]^T$, is the vector that connects \{i\} to \{b\}, through which the trajectory to be tracked is defined as $\prescript{i}{}{\mathbf{p}_d(t)} = \left[\,x_d(t)\:\:\:\:z_d(t)\right]^T$. The coordinate transformation from \{b\} to \{i\} is given by the following rotation matrix $\mathbf{R} \in SO(2)$:
\begin{equation}\label{eq:rotation}
    \mathbf{R} = \begin{bmatrix}
        \cos{\theta} & \sin{\theta}\\
        -\sin{\theta} & \cos{\theta}
    \end{bmatrix}\,,
\end{equation}
where the rotation angle $\theta$ is defined as the pitch angle. The inverse coordinate transformation, from \{i\} to \{b\}, is simply given by ${\mathbf{R}}^T$.
\subsection{External Forces and Moments} 

A launch vehicle experiences three main sources of external forces and moments during flight: gravitational, aerodynamic, and propulsive.
Considering the Earth as a perfect sphere, and looking at the definition of the inertial frame \{i\}, the gravity force is simply
\begin{equation}\label{gravity}
\prescript{i}{}{\mathbf{f_g}} = \left[\,-mg \:\:\:\: 0\,\right]^T\,,
\end{equation}
where $m$ is the instantaneous mass of the vehicle and $g$, the gravitational acceleration, varies with the altitude according to 
\begin{equation}\label{eq:gravity}
g=g_0\,R_E^2\,/\,(R_E+x)^2\,,
\end{equation}
where $R_E$ is the mean Earth radius and $g_0$ is the gravitational acceleration at surface level.

The magnitude and direction of the thrust vector, $\prescript{b}{}{\mathbf{f_p}}$, are the adjustable inputs that enable active trajectory tracking control. Looking at Fig. \ref{fig:frames}, the thrust vector expressed in \{b\} is given by
\begin{equation}
    \prescript{b}{}{\mathbf{f_p}} = \left[\,\prescript{b}{}{f_{p_x}} \:\:\:\: \prescript{b}{}{f_{p_z}}\,\right]^T =\left[\, T\cos{\mu}\:\:\:\: T\sin{\mu}\,\right]^T\,,
\end{equation}
where $T$ is the thrust magnitude and $\mu$ is the engine's gimbal angle. Upon engine deflection, the thrust vector produces a pitching moment which is used to control the rotation of the vehicle, given by
\begin{equation}
    \tau_p =  \prescript{b}{}{f_{p_z}}\,l = T\sin{\mu}\,l\,,
\end{equation}
where $l$ is the distance between the gimbal point and the center of mass of the vehicle, $x_{cm}$, both measured from the top.

The generation of thrust causes the depletion of propellant. Assuming ideal propulsion and a constant vacuum specific impulse $I_{\text{sp}_0}$, the propellant depletion rate, or mass flow rate, is given by \cite{successive_conv}
\begin{equation}\label{eq:massflow}
    \dot{m} = -\frac{T}{I_{\text{sp}_0}\,g_0}-\frac{p_a\,A_e}{I_{\text{sp}_0}\,g_0}\,,
\end{equation}
where the second term accounts for back-pressure drops and depends on the atmospheric pressure $p_a$ and nozzle exit area, $A_e$.

The motion of the vehicle through the fluid composing the atmosphere causes the appearance of aerodynamic loads. The aerodynamic force expressed in \{b\} can be modeled as
\begin{equation}
    \prescript{b}{}{\mathbf{f_a}} = \left[\,\prescript{b}{}{f_{a_x}} \:\:\:\: \prescript{b}{}{f_{a_z}}\,\right]^T = \left[\, -\overline{q}\,C_AS \:\:\:\: -\overline{q}\,C_NS\,\right]^T\,,
\end{equation}
where $\overline{q}$ is the dynamic pressure, and $C_A$ and $C_N$ are the axial and normal aerodynamic force coefficients, respectively, and $S$ is a reference area, typically the fuselage cross-sectional area. The force coefficients can be related to the Lift and Drag coefficients, $C_L$ and $C_D$, using the angle of attack, $\alpha$:
\begin{subequations}
\begin{equation}
C_A = C_D\cos{\alpha} - C_L\sin{\alpha}\,,
\end{equation}
\begin{equation}
C_N = C_L\cos{\alpha} + C_D\sin{\alpha}\,,
\end{equation}
\end{subequations}
with the angle of attack given by $\alpha = atan_2(w_\text{rel}, u_\text{rel})$, where $u_\text{rel}$ and $w_\text{rel}$ are the components of the fluid relative velocity vector expressed in the body frame, $\prescript{b}{}{\mathbf{v_{rel}}}$. This vector is computed through $\prescript{b}{}{\mathbf{v_{rel}}} = \prescript{b}{}{\mathbf{v}} - \prescript{b}{}{\mathbf{v_{w}}}$, where $\prescript{b}{}{\mathbf{v}}$ is the linear velocity vector and $\prescript{b}{}{\mathbf{v_{w}}}$ is the wind velocity vector, both expressed in \{b\}. The Lift and Drag coefficients can be stored as a function of the angle of attack and Mach number.

For the aerodynamic pitching moment, it is assumed that its only cause is the offset between the center of pressure, $x_{cp}$, where the aerodynamic forces are applied, and the center of mass. Hence, it is modeled as
\begin{equation}
    \tau_a = \prescript{b}{}{f_{a_z}}\,SM\,\overline{d}\,,
\end{equation}
where $\overline{d}$ is a reference length, usually the maximum diameter of the fuselage, and $SM = (x_{cp}- x_{cm})/\overline{d}$ is the static stability margin.
\subsection{Equations of Motion}

With the reference frames defined and the external forces and moments characterized, it is possible to obtain the pitch plane rigid-body equations of motion. Following the Newton-Euler formalism, and neglecting the impact of the mass, center of mass, and inertia (MCI) time derivatives and of moving masses (including “tail-wags-dog” moment and rocket jet
damping), the pitch plane rigid-body equations of motion are given by
\begin{equation}\label{eq:position}
    m\,\prescript{i}{}{\ddot{\mathbf{p}}} = -m\,g\,\mathbf{e_1}+\mathbf{R}\,\prescript{b}{}{\mathbf{f_p}} + \prescript{i}{}{\mathbf{f_a}} \,,
\end{equation}
\begin{equation}\label{eq:attitude}
   j_y\,\dot{q} = \tau_p + \tau_a\,,
\end{equation}
\begin{equation}\label{eq:attitudekin}
   \dot{\theta} = q\,,
\end{equation}
where $j_y$ is the transverse moment of inertia and $q$ is the angular velocity or pitch rate.

Taking the propulsive vector in the body frame as the control input, i.e.,  $\mathbf{u} = \prescript{b}{}{\mathbf{f_p}}$, and considering the inertial position and orientation as the outputs, the system can be written in state-space form as follows: 
\begin{equation}\label{eq:states}
      {\mathbf{x}}=\left[\,x\:\:\:\:\dot{x}\:\:\:\:z\:\:\:\:\dot{z}\:\:\:\:\theta\:\:\:\:q \,\right]^T\,,
\end{equation}
\begin{equation}\label{eq:sys_xdot}
\renewcommand{\arraystretch}{1.5}
    \dot{\mathbf{x}} = 
    \begin{bmatrix}
        \dot{x}\\
         -g + \frac{\prescript{i}{}{f_{a_x}}}{m}\\
         \dot{z}\\
         \frac{\prescript{i}{}{f_{a_z}}}{m}\\
        q\\
        \frac{\tau_a}{j_y} 
    \end{bmatrix}+
    \begin{bmatrix}
        0 & 0\\
        \frac{\cos{\theta}}{m} & \frac{\sin{\theta}}{m}\\
        0 & 0\\
        -\frac{\sin{\theta}}{m} & \frac{\cos{\theta}}{m}\\
        0 & 0\\
        0 & \frac{l}{j_y}
    \end{bmatrix}\,\mathbf{u}\,,
\end{equation}
\begin{equation}\label{eq:sys_y}
     {\mathbf{y}} = \left[\,x\:\:\:\:z\:\:\:\:\theta\,\right]^T\,.
\end{equation}

\section{Control}\label{sec:control}
To tackle the trajectory tracking control problem, an inner-outer loop structure is adopted (Fig.~\ref{fig:gen_arch}), in which the inner loop comprises the angular and vertical motion ($\theta$ and $x$) while the outer loop entails the downrange motion ($z$). As opposed to the traditional formulation, in which trajectory tracking is addressed by dividing the problem into attitude control and guidance, our approach relies on the inner-outer control structure to exploit the full dynamics of the system and obtain almost-global stability guarantees.
\begin{figure*}
\centering
\includegraphics{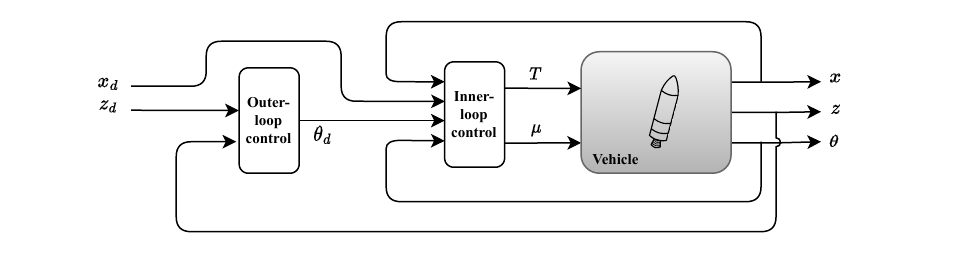}
\caption{\bf{High-level inner-outer control architecture.}}
\label{fig:gen_arch}
\end{figure*}

Feedback linearization is used to cancel out the non-linearities of the derived model and reduce the dynamics to a chain of integrators, with the aerodynamic forces and moment being estimated in real-time, following an adaptive backstepping methodology. The adaptation scheme does not require a model for the aerodynamics of the vehicle, which is a considerable source of uncertainty , partially addressing the inherent drawback known to feedback linearization of requiring an accurate knowledge of the plant dynamics and parameters.

\subsection{Feedback Linearization}

Feedback Linearization is a nonlinear control approach that aims to algebraically transform nonlinear dynamics of systems, through nonlinear change
of coordinates and nonlinear state feedback, into a
model that is linear in the new set of coordinates. The
linear model produced is an exact representation of the
original nonlinear model over a large set of operating
points \cite{Martins}. Given a nonlinear system of the form
\begin{equation}\label{eq:xdot}
    \dot{\mathbf{x}} = \mathbf{f(x)} + \mathbf{g(x)}\,\mathbf{u}\,,
\end{equation}
\begin{equation} \label{eq:y}
    \mathbf{y} = \mathbf{h(x)}\,,
\end{equation}
where ${\mathbf{f}} \in {\mathbb{R}}^n$, ${\mathbf{g}} \in {\mathbb{R}}^{n\times m}$, and ${\mathbf{h}} \in {\mathbb{R}}^m$ are sufficiently smooth nonlinear vector fields in a domain $D \subset {\mathbb{R}}^n$. According to Isidori \cite{Isidori}, if the system has a vector relative degree ${\mathbf{r}} = \{r_1,...\,,r_m\}$, with the sum of the components equal to its order, $r_1 +r_2+...+r_m = n$, then a fully linear and decoupled version can be obtained through the application of the following coordinate
transformation, or \textit{diffeomorphism}, $\boldsymbol\Phi({\mathbf{x}}) \in {\mathbb{R}}^n$:
\begin{equation}\label{eq:diffeomorphism}
\boldsymbol{\xi}_{j,\,k} = \boldsymbol{\phi}_{j,\,k} = {\mathcal{L}}^{k-1}_{{\mathbf{f}}}{\mathbf{h}}_j(\mathbf{x})\,,
\end{equation}
with $k \in \{1,...\,,r_j\}$ and where $\mathcal{L}$ represents the Lie derivative, and the nonlinear static state feedback control law
\begin{equation}\label{eq:input_feedback}
\mathbf{u} = -\boldsymbol{\Lambda}^{-1}(\mathbf{x})\,\mathbf{b(x)} + \boldsymbol{\Lambda}^{-1}(\mathbf{x})\,\mathbf{v}\,,
\end{equation}
where $\mathbf{v} \in {\mathbb{R}}^n$ is the vector of the transformed control inputs, $\boldsymbol{\Lambda}({\mathbf{x}}) \in {\mathbb{R}}^{m\times m}$ is the \textit{decoupling matrix} given by
\begin{equation}
    \boldsymbol{\Lambda}(\mathbf{x}) = 
    \begin{bmatrix}
        {\mathcal{L}}_{{\mathbf{g}}_1}{\mathcal{L}}^{r_1-1}_{\mathbf{f}}{\mathbf{h}}_1(\mathbf{x}) & \dots & {\mathcal{L}}_{{\mathbf{g}}_m}{\mathcal{L}}^{r_1-1}_{\mathbf{f}}{\mathbf{h}}_1(\mathbf{x})\\
        \vdots & \ddots & \vdots\\
        {\mathcal{L}}_{{\mathbf{g}}_1}{\mathcal{L}}^{r_m-1}_{\mathbf{f}}{\mathbf{h}}_m(\mathbf{x}) & \dots & {\mathcal{L}}_{{\mathbf{g}}_m}{\mathcal{L}}^{r_m-1}_{\mathbf{f}}{\mathbf{h}}_m(\mathbf{x})
    \end{bmatrix}\,,
\end{equation}
and $\mathbf{b(x)} \in {\mathbb{R}}^m$ is given by
\begin{equation}
    {\mathbf{b(x)}} = \left[\,{\mathcal{L}}^{r_1}_{\mathbf{f}}{\mathbf{h}}_1({\mathbf{x}})\:\:\dots\:\:{\mathcal{L}}^{m}_{\mathbf{f}}{\mathbf{h}}_m({\mathbf{x}})\,\right]^T\,.
\end{equation}\\
From (\ref{eq:input_feedback}), we observe that the \textit{decoupling matrix}, $\boldsymbol{\Lambda}(\mathbf{x})$, must be non-singular. With the application of the input-output feedback linearization, the resulting model is linear in the new set of coordinates and is composed by a set of $m$ chains of $r_j$ integrators with transformed inputs $\mathbf{v}_j$:
\begin{equation}\label{eq:intchain}
\begin{split}
        \boldsymbol{\dot{\xi}}_{j,1} &= \boldsymbol{\xi}_{j,2}\\
    &\vdots\\
    \boldsymbol{\dot{\xi}}_{j,r_j-1} &=\boldsymbol{\xi}_{j,r_j} \\
    \boldsymbol{\dot{\xi}}_{j,r_j} &= {\mathbf{v}}_j\,.
    \end{split}
\end{equation}

\subsection{Inner-loop}

In the majority of the works found in the literature, the
inner-loop, commonly referred to as attitude control, solely governs the attitude dynamics. In this work,
this loop also encapsulates the altitude dynamics since
this inclusion does not yield additional singularities and allows to assign the computation of all physical inputs - thrust magnitude and deflection angle - to the inner-loop. Let ${\mathbf{x_{in}}} \in {\mathbb{R}}^2$ be the vector of the state variables that compose the inner dynamics written in the tracking error coordinates
\begin{equation} \label{eq:innerstates}
    {\mathbf{x_\textbf{in}}}=\left[\,e_x\:\:\:\:\dot{e}_x\:\:\:\:e_\theta\:\:\:\:\dot{e}_\theta \,\right]^T\,,
\end{equation}
where $e_x = x-x_d$ and $e_\theta = \theta-\theta_d$ are the inner-loop tracking errors, and let ${\mathbf{u_{in}}} \in {\mathbb{R}}^2$ be the vector of the inner control inputs
\begin{equation}  \label{eq:inneru}
{\mathbf{u_\textbf{in}}}=\prescript{b}{}{\mathbf{f_p}}=\left[\, T\,\cos{\mu}\:\:\:\: T\,\sin{\mu}\,\right]^T.
\end{equation}
Considering the complete dynamics in (\ref{eq:states})-(\ref{eq:sys_y}), the inner error dynamics can be written in the form described in (\ref{eq:xdot})-(\ref{eq:y}):
\begin{equation}\label{eq:inner_xdot}
\renewcommand{\arraystretch}{1.5}
    \dot{\mathbf{x}}_\textbf{in} =
    \begin{bmatrix}
          \dot{e}_x\\
         -g + \frac{\prescript{i}{}{f_{a_x}}}{m}-\ddot{x}_d\\
        \dot{e}_\theta\\
        \frac{\tau_a}{j_y} - \ddot{\theta}_d
    \end{bmatrix}+
    \begin{bmatrix}
        0 & 0\\
        \frac{\cos{\theta}}{m} & \frac{\sin{\theta}}{m}\\
        0 & 0\\
        0 & \frac{l}{j_y}
    \end{bmatrix}\,\mathbf{u_\textbf{in}}\,,
\end{equation}
\begin{equation}\label{eq:inner_y}
     {\mathbf{y_\textbf{in}}} = \left[\,e_x\:\:\:\:e_\theta\,\right]^T\,.
\end{equation}
\begin{proposition}\label{prop:rel_deg}
    The inner-loop error dynamics described by Eq. (\ref{eq:inner_xdot}) with output given by (\ref{eq:inner_y}) have a well-defined vector relative degree ${\mathbf{r_{in}}}=\{2,2\}$ on the set $\{{\mathbf{x_\textbf{in}}}\in{\mathbb{R}}^4:|\theta|<\pi/2\}$.
\end{proposition}
\begin{proof}
Proof in Appendix \ref{app:proof1}.
\end{proof}  

From Proposition \ref{prop:rel_deg}, it follows that the inner-loop tracking error dynamics, defined by the system (\ref{eq:inner_xdot})-(\ref{eq:inner_y}), meet the necessary conditions to be input-output feedback linearized through the application of a coordinate transformation as defined in (\ref{eq:diffeomorphism}) and the static state feedback law
\begin{equation}\label{eq:inner_lin}
    \mathbf{u_\textbf{in}} = -\boldsymbol{\Lambda_{\textbf{in}}}^{-1}(\mathbf{x_\textbf{in}})\,\mathbf{b_\textbf{in}(x_\textbf{in})} + \boldsymbol{\Lambda_{\textbf{in}}}^{-1}(\mathbf{x_\textbf{in}})\,\mathbf{v_\textbf{in}}\,,
\end{equation}

with 
\begin{equation}\label{eq:inner_diff}
\renewcommand{\arraystretch}{1.5}
    \boldsymbol{\Lambda}_{\textbf{in}}(\mathbf{x_\textbf{in}}) = 
    \begin{bmatrix}
        \frac{\cos{\theta}}{m} & \frac{\sin{\theta}}{m}\\
        0 & \frac{l}{j_y}
    \end{bmatrix}\,,
\end{equation}
and
\begin{equation}\label{eq:inner_b}
{\mathbf{b_\textbf{in}(x_\textbf{in})}}=\left[\,-g+\frac{\prescript{i}{}{fa}_x}{m}-\ddot{x}_d\:\:\:\:\frac{\tau_a}{j_y}-\ddot{\theta}_d\,\right]^T\,.
\end{equation}
For this system, the application of the \textit{diffeomorphism} defined by (\ref{eq:attitudekin}) preserves the original coordinates, and the resulting linear system is given by two chains of integrators, which can be written in state-space form as
\begin{equation}
    \boldsymbol{\xi_\textbf{in}} = {\mathbf{x_\textbf{in}}} = \left[\,e_x\:\:\:\:\dot{e}_x\:\:\:\:e_\theta\:\:\:\:\dot{e}_\theta \,\right]^T\,,
\end{equation}
\begin{equation}\label{eq:inner_xidot}
    \boldsymbol{\dot{\xi}_\textbf{in}} =
    \begin{bmatrix}
        0 & 1 & 0 & 0\\
        0 & 0 & 0 & 0\\
        0 & 0 & 0 & 1\\
        0 & 0 & 0 & 0
    \end{bmatrix}\,\boldsymbol{\xi_\textbf{in}}+
    \begin{bmatrix}
        0 & 0\\
        1 & 0\\
        0 & 0\\
        0 & 1
    \end{bmatrix}\,\mathbf{v_\textbf{in}}\,,
\end{equation}
\begin{equation}\label{eq:inner_xiy}
    { \mathbf{y_\textbf{in}}} = \left[\,e_x\:\:\:\:e_\theta\,\right]^T\,.
\end{equation}
 With the inner-loop error dynamics reduced to two decoupled double integrators, designing a stable state feedback law for the vector of transformed control inputs $\mathbf{v_{in}}$ is straightforward, with many well-established linear domain design methods available for that purpose. However, the feedback linearization law requires an accurate knowledge of system parameters, namely its MCI properties, and of the external forces and moments acting on the vehicle, as seen in (\ref{eq:inner_diff})-(\ref{eq:inner_b}). Potential mismatches between the model and the real system will cause a degradation of the double integrator dynamics and may lead to instability. In our work, we have decided to assume an accurate knowledge of the gravitational acceleration and of the MCI properties of the vehicle, while estimating in real-time the aerodynamic forces and moment acting on the vehicle with a model-free adaptive approach.

 Firstly, let us respectively define $\hat{f}_{a_x}$ and $\hat{\tau}_a$ as the estimates on the vertical aerodynamic force and on the aerodynamic pitching moment. By using the estimates directly on the input-output feedback linearization (\ref{eq:inner_lin}) with
 \begin{equation}
 {\mathbf{b_{in}(x_{in})}}=\left[\,-g+\frac{\prescript{i}{}{\hat{f}a}_x}{m}-\ddot{x}_d\:\:\:\:\frac{\hat{\tau}_a}{j_y}-\ddot{\theta}_d\,\right]^T\,,
 \end{equation}
 the state-space representation of the linear system in (\ref{eq:inner_xidot}) reshapes into
 \begin{equation}\label{eq:xi_in_perturbed}
     \boldsymbol{\dot{\xi}_\textbf{in}} =
    \begin{bmatrix}
        0 & 1 & 0 & 0\\
        0 & 0 & 0 & 0\\
        0 & 0 & 0 & 1\\
        0 & 0 & 0 & 0
    \end{bmatrix}\,\boldsymbol{\xi_\textbf{in}}+
    \begin{bmatrix}
        0 & 0\\
        1 & 0\\
        0 & 0\\
        0 & 1
    \end{bmatrix}\,\left(\mathbf{v_\text{in}}+
    \begin{bmatrix}
       \frac{\prescript{i}{}{\tilde{f}_{a_x}}}{m}\\[5pt]
       \frac{{\tilde{\tau}_a}}{j_y}
    \end{bmatrix}\right)\,,
 \end{equation}
in which $\tilde{f}_{a_x} = f_{a_x} -\hat{f}_{a_x}$ and $\tilde{\tau}_a = \tau_a - \hat{\tau}_a$ are the estimation errors. Looking at system (\ref{eq:xi_in_perturbed}), we see that, by considering an error between the aerodynamic estimates and the true values, the input-output linearized system can be seen as a pair of double integrators with an external perturbation that depends on the estimation errors. 

Considering the decoupling of the system, the control law for the transformed inputs can be designed independently for each chain of integrators, which have a similar generic structure given by 
\begin{equation}\label{eq:gen_lin}
    \dot{\xi}_1 = \xi_2\,,\hspace{10pt}
    \dot{\xi}_2 = v+\tilde{d}\,,
\end{equation}
where $v$ is the transformed control input and $\tilde{d}$ represents the external disturbance to the double integrator, given by the estimation error. In order to design a stable controller that is able to reject the disturbance by estimating it, and thus obtain adaptation laws for the aerodynamic quantities used in feedback linearization, the nonlinear adaptive backstepping method \cite{kokotovic} is used.

\begin{lemma}\label{eq:lemma_inner}
    Given the linear system (\ref{eq:gen_lin}), let $k_1$, $k_2$, and $\gamma$ be constant positive gains, the feedback law
\begin{equation}\label{eq:feedback_law}
    v = -\left(1+k_1k_2\right)\xi_1-\left(k_1+k_2\right)\xi_2\,,
\end{equation}
 in conjunction with the following dynamics for the disturbance
 \begin{equation}\label{eq:adapt_law}
     \dot{\tilde{d}} = -\gamma\left(k_1\xi_1 +\xi_2\right)\,,
 \end{equation}
yield a globally asymptotically stable origin of the closed-loop system.
 
\end{lemma}
\begin{proof}
Proof in appendix \ref{app:lemma1}.
\end{proof}

Using Lemma \ref{eq:lemma_inner}, more precisely Eq. (\ref{eq:feedback_law}), the feedback law for the virtual input $\mathbf{v_\textbf{in}}$ can be defined as 
\begin{equation}\label{eq:virtual_input_law}
    \mathbf{v_\textbf{in}}=\begin{bmatrix}
         -\left(1+k_{1_x}k_{2_x}\right)e_x-\left(k_{1_x}+k_{2_x}\right)\dot{e}_x\\[5pt]
         -\left(1+k_{1_\theta}k_{2_\theta}\right)e_\theta-\left(k_{1_\theta}+k_{2_\theta}\right)\dot{e}_\theta
    \end{bmatrix}\,,
\end{equation}
where  $k_{1_x}$, $k_{2_x}$, $k_{1_\theta}$, and $k_{2_\theta}$ are the control gains. Similarly, using Lemma \ref{eq:lemma_inner}, particularly Eq. (\ref{eq:adapt_law}), adaptation laws can be derived for the vertical aerodynamic force, $\prescript{i}{}{\hat{f}}_{a_x}$, and for the aerodynamic moment, $\hat{\tau}_a$. Under the assumption that the true values are slowly time-varying, i.e., $\prescript{i}{}{\dot{f}}_{a_x} \approx 0$ and $\dot{\tau}_a \approx 0$, we have that
\begin{equation}
    \prescript{i}{}{\dot{\hat{f}}}_{a_x} = -\tilde{f}_{a_x}\frac{\dot{m}}{m} +\gamma_xm\left(k_{1_x}e_x+\dot{e}_x\right)\,,
\end{equation}
\begin{equation}
        \dot{\hat{\tau}}_a = -\tilde{\tau}_{a}\frac{\dot{j}_y}{j_y} +\gamma_\theta\,j_y\left(k_{1_\theta}e_\theta+\dot{e}_\theta\right)\,,
\end{equation}
where $\gamma_x$ and $\gamma_\theta$ are the adaptation gains. Neglecting the mass and inertia time derivatives, i.e, considering the reasonable assumption that $\dot{m}/m \ll 1$ and $\dot{j}_y/j_y \ll 1$ the following adaptations law are obtained:
\begin{equation}\label{eq:adaptlaw_fa}
     \prescript{i}{}{\hat{f}}_{a_x}(t) = \prescript{i}{}{\hat{f}}_{a_x}(0) +\gamma_x\int_{0}^{t} m\left(k_{1_x}e_x+\dot{e}_x\right) \,d\tau\,,
\end{equation}
\begin{equation}\label{eq:adaptlaw_ta}
      \hat{\tau}_a(t) = \hat{\tau}_a(0) +\gamma_\theta\int_{0}^{t} j_y\left(k_{1_\theta}e_\theta+\dot{e}_\theta\right) \,d\tau\,.
\end{equation}

\begin{theorem}\label{T:inner} Let the inner-loop error dynamics be described by Eqs. (\ref{eq:innerstates}) and (\ref{eq:inner_xdot}), and have its output equation given by (\ref{eq:inner_y}). Under slowly time-varying aerodynamic loads and negligible mass and inertia time derivatives, the closed-loop system that results from applying the feedback law defined by (\ref{eq:inner_lin}) using the virtual input (\ref{eq:virtual_input_law}) and the real-time estimates given by (\ref{eq:adaptlaw_fa})-(\ref{eq:adaptlaw_ta}) is regionally asymptotically stable, with the stability region defined by $\{{\mathbf{x_\textbf{in}}}\in{\mathbb{R}}^4:|\theta|<\pi/2\}$.
\end{theorem}
\begin{proof}
    From Proposition \ref{prop:rel_deg}, the input-output feedback linearization is valid on the set $\{{\mathbf{x_\textbf{in}}}\in{\mathbb{R}}^4:|\theta|<\pi/2\}$. By applying the input-output linearization through the feedback law defined by (\ref{eq:inner_lin}) with the real-time estimates given by (\ref{eq:adaptlaw_fa}) and (\ref{eq:adaptlaw_ta}), the resulting closed-loop dynamics are of the type (\ref{eq:xi_in_perturbed}). Under the stated assumptions, the real-time estimates yield equivalent estimation error dynamics and virtual input (\ref{eq:virtual_input_law}) to the ones defined in Lemma \ref{eq:lemma_inner}. Hence, the origin of the closed-loop error dynamics is asymptotically stable on the set $\{{\mathbf{x_\textbf{in}}}\in{\mathbb{R}}^4:|\theta|<\pi/2\}$, yielding regionally asymptotically stable closed-loop inner dynamics.
    
\end{proof}

From Theorem \ref{T:inner}, we have that the inner-loop dynamics, concerning the vertical and angular motion, are now almost-globally stabilized, with the only non-stable region happening for an horizontal orientation, given the way the reference frames were defined. Considering the class of vehicles under study, this orientation is deemed outside of the nominal range of operation, meaning that the derived inner-loop controller is stable in the region of interest. Moreover, looking at the input-output feedback linearization given by (\ref{eq:inner_lin}) and the virtual control input in (\ref{eq:virtual_input_law}), it is possible to conclude that the position and orientation references given to the inner-loop controller shall have well-defined second derivatives to serve as feedforward control inputs. Insofar as this condition is met, the inner-loop controller is able to track generic references.  

Although slowly time-varying aerodynamic loads have been assumed to derive the adaptation laws and prove the stability of the resulting closed-loop system, the dynamics of the estimation process can be made arbitrarily fast through the adjustment of the adaptation gains, allowing to track the time-varying aerodynamic loads that the vehicle is subject to due to changing environmental conditions and velocity.

Finally, it is necessary to retrieve the thrust vector magnitude, $T$, and gimbal angle, $\mu$, from the commanded inner-loop input, $\mathbf{u}_\textbf{in}$. Looking at the input definition in (\ref{eq:inneru}), it is straightforward to obtain
\begin{subequations}
\begin{equation}
    \mu = \arctan_2\left(\frac{{ u_{\text{in}_2}}}{{ u_{\text{in}_1}}}\right)\,,
\end{equation}
\begin{equation}
  T = \frac{{ u_{\text{in}_1}}}{\cos{\mu}}\,.
\end{equation}
\end{subequations} 

 \subsection{Outer-loop}

 With the inner-loop stabilized, a feedback law for the outer-loop shall be designed with the objective of controlling the downrange motion of the vehicle, $z$, by providing appropriate attitude references, $\theta_d$, to the inner-loop controller. To do so, the outer-loop dynamics must first be characterized, considering the closed-loop stabilized inner dynamics. In this work, and similarly to what was presented for quadrotors in \cite{Martins}, we have decided to control the horizontal motion through the \textit{zero dynamics} of the inner-stabilized system.

 According to \cite{Isidori}, to obtain the \textit{zero dynamics} of the system, we shall solve the \textit{Problem of zeroing the output}. It is important to note that a zero output of the inner-stabilized system corresponds to a null tracking error on the vertical position and orientation. For $\mathbf{y_\textbf{in}} = \mathbf{0}$, it is straightforward to verify that $\boldsymbol{\xi}_\textbf{in}=\mathbf{0}$ and, from (\ref{eq:intchain}), $\mathbf{v_\textbf{in}}=\mathbf{0}$.  Given (\ref{eq:input_feedback}), the input vector $\mathbf{u^{*}_\textbf{in}}$ that solves the \textit{Problem of zeroing the output} is obtained from
 \begin{equation}
     \mathbf{0} = \mathbf{b_\textbf{in}(x^{*}_\textbf{in})} + \boldsymbol{\Lambda_{\textbf{in}}}(\mathbf{x^{*}_\textbf{in}})\,\mathbf{u^{*}_\textbf{in}}\,,
 \end{equation}
yielding
 \begin{equation}\label{eq:zero_input}
 \renewcommand{\arraystretch}{2}
 {\mathbf{u^{*}_\textbf{in}}} = \begin{bmatrix}
        \frac{1}{\cos{\theta}}\left( mg -\prescript{i}{}{\hat{f}_{a_x}}+m\ddot{x}_d +\frac{\sin{\theta}}{l}\left(\hat{\tau}_a -j_y\ddot{\theta}_d \right)\right)\\
          l^{-1}\left(j_y\ddot{\theta}_d - \hat{\tau}_a\right) 
     \end{bmatrix}\,.
 \end{equation}
 Looking at the complete system, described in Eqs. (\ref{eq:states})-(\ref{eq:sys_y}), the outer-loop dynamics are given by
 \begin{equation}\label{eq:outer_dynamics}
     \ddot{z} = \frac{\prescript{i}{}{f_{a_z}}}{m} - \frac{\sin{\theta}}{m}{ u_{\text{in}_1}} +  \frac{\cos{\theta}}{m}{u_{\text{in}_2}}\,,
 \end{equation}
 where $ u_{\text{in}_1}$ and $u_{\text{in}_2}$ are the components of the input vector $\mathbf{u_\textbf{in}}$. By replacing input vector (\ref{eq:zero_input}) in the outer-loop dynamics (\ref{eq:outer_dynamics}), the \textit{zero dynamics} describing the horizontal motion in the presence of stabilized inner dynamics are obtained:
 \begin{equation}\label{eq:zero_dynamics}
          \ddot{z} = \frac{\prescript{i}{}{f_{a_z}}}{m} + a\,\tan{\theta} +  \frac{b}{\cos{\theta}}\,,
\end{equation}
where
 \begin{equation*}
          a = -g+\frac{\prescript{i}{}{f_{a_x}}}{m}-\ddot{x}_d\,,\hspace{10pt} b = \frac{j_y\ddot{\theta}_d-\tau_a}{m\,l}\,.
\end{equation*}
An outer-loop controller is required to control these dynamics through the pitch angle $\theta$. Once again, the dynamics will be written in the error coordinates. Let
\begin{equation}\label{eq:outerstates}
    {\mathbf{x_\textbf{out}}}=\left[\,e_z\:\:\:\:\dot{e}_z\,\right]^T
\end{equation}
denote the state vector for the outer dynamics and
\begin{equation}   \label{eq:outerinput} 
     u_\text{out}= a\,\tan{\theta} +  \frac{b}{\cos{\theta}}
\end{equation}
be the outer-loop control input. The outer error dynamics can be written in the form (\ref{eq:xdot})-(\ref{eq:y}):
\begin{equation}\label{eq:outer_xdot}
\renewcommand{\arraystretch}{2}
    \dot{\mathbf{x}}_\textbf{out} = \begin{bmatrix}
         \dot{e}_z\\
         \frac{\prescript{i}{}{f_{a_z}}}{m}-\ddot{z}_d
    \end{bmatrix}+
    \begin{bmatrix}
        0\\
        1
    \end{bmatrix}\,u_\text{out}\,,
\end{equation}
\begin{equation} \label{eq:outer_y}
    y_\text{out} =  e_z\,.
\end{equation}

\begin{proposition}\label{prop:out}
   The outer-loop error dynamics described by Eq. (\ref{eq:outer_xdot}) and with output given by (\ref{eq:outer_y}) have a well-defined relative degree $r_\text{out} = 2$ for any $\mathbf{x_\textbf{out}}$.
\end{proposition}
\begin{proof}
   Proof in Appendix \ref{app:prop2}.
\end{proof}

From Proposition \ref{prop:out}, and in accordance with \cite{Khalil}, the nonlinear system (\ref{eq:outer_xdot})-(\ref{eq:outer_y}) is input-output linearizable. More specifically, the application of the SISO version of the \textit{diffeomorphism} in (\ref{eq:diffeomorphism})
\begin{equation}
\renewcommand{\arraystretch}{1.5}
     \boldsymbol{\xi}_\textbf{out}=\boldsymbol{\phi}_\textbf{out} = \begin{bmatrix}
         h(\mathbf{x_\textbf{out}})\\
        { \mathcal{L}}_{f}h(\mathbf{x_\textbf{out}})
     \end{bmatrix}=
     \begin{bmatrix}
         e_z\\
         \dot{e}_z
     \end{bmatrix} = \mathbf{x_\textbf{out}}\,,
\end{equation}
and of the static state feedback law
\begin{equation} \label{eq:outer_lin}
\begin{split}
    u_\text{out} &= \left({\mathcal{L}}_{g}{\mathcal{L}}_{f}h(\mathbf{x_\textbf{out}})\right)^{-1}\,\left(-{\mathcal{L}}^2_{f}h(\mathbf{x_\textbf{out}})+v_\text{out}\right) =\\ &=\ddot{z}_d-\frac{\prescript{i}{}{f_{a_z}}}{m} + v_\text{out}\,,
    \end{split}
\end{equation}
yields the input-output linearized version of the system, described by a double integrator in the transformed coordinates:
\begin{equation}
\renewcommand{\arraystretch}{1.5}
    \boldsymbol{\xi_\textbf{out}} = {\mathbf{x_\textbf{out}}} = \left[\,e_z\:\:\:\:\dot{e}_z\,\right]^T\,,
\end{equation}
\begin{equation}\label{eq:outer_xidot}
    \boldsymbol{\dot{\xi}_\textbf{out}} =
    \begin{bmatrix}
        0 & 1\\
        0 & 0
    \end{bmatrix}\,\boldsymbol{\xi_\textbf{out}}+
    \begin{bmatrix}
        0\\
        1
    \end{bmatrix}\,v_\text{out}\,,
\end{equation}
\begin{equation}\label{eq:outer_xiy}
     y_\text{out} = e_z\,.
\end{equation}

Similarly to what was obtained for the inner-loop, the outer-loop dynamics are now reduced to a double integrator with virtual input $v_\text{out}$ - a linear system for which many stabilizing static state feedback laws are available. However, for the outer-loop, given that the \textit{zero dynamics} are not the exact dynamics of the system, it is not straightforward to add an equivalent adaptation layer to estimate the horizontal aerodynamic force, $\prescript{i}{}{f_{a_z}}$, to be used in the input-output feedback law (\ref{eq:outer_lin}). Instead, the inner-loop aerodynamic estimates are used to obtain an indirect estimate on that quantity, based on first principles. Following the aerodynamics modeling presented in Section \ref{sec:physical_model}, the indirect estimate on the horizontal aerodynamic force is
\begin{equation}\label{eq:outer_estimate}
    \prescript{i}{}{\hat{f}_{a_z}} = \frac{\cos{\theta}\,\hat{\tau}_a}{SM\,d} - \tan{\theta}\left(\prescript{i}{}{\hat{f}}_{a_x}-\frac{\sin{\theta}\,\hat{\tau}_a}{SM\,d}\right)\,,
\end{equation}
which requires the additional knowledge of the evolution of the position of the center of pressure, $x_{cp}$, which can be stored as function of the angle of attack, $\alpha$, and the velocity in Mach, $M$, to determine the static margin, $SM$. The expression for the indirect estimate is non-singular in the same region as the inner-loop controller in terms of pitch angle, with the addition of a singularity for a null static margin. In any case, the static margin is expected to be negative due to the natural instability of the vehicle and, if null, the singularity can be dealt with as it can only happen instantaneously given the dynamic behavior of the system. Considering this estimate, the input-output linearization then reshapes into
\begin{equation}\label{eq:outer_lin_hat}
    u_\text{out}  = \ddot{z}_d - \frac{\prescript{i}{}{\hat{f}_{a_z}}}{m} + v_\text{out}\,.
\end{equation}
With the input-output linearization completely defined, the Linear Quadratic Regulator (LQR) \cite{friedland}, with additional integral action to deal with steady state-error and to mitigate the effect of perturbations \cite{dossantos}, is used to obtain the feedback law for the outer-loop virtual input $v_\text{out}$. The outer-loop virtual input is defined as
\begin{equation}\label{eq:v_out}
    v_\text{out} = -k_z\,e_z - k_{\dot{z}}\,\dot{e}_z + -k_i\,\zeta_z\,,
\end{equation}
where $k_z$, $k_{\dot{z}}$, and $k_i$ are the gains that result from applying the LQR computation to the fedforward double integrator error dynamics with the additional integral state $\zeta_z$, which satisfies
\begin{equation}
    \dot{\zeta}_z = e_z\,.
\end{equation}
The closed-loop error dynamics that result from applying the input-output linearization in (\ref{eq:outer_lin_hat}) with the horizontal aerodynamic force estimate given by Eq. (\ref{eq:outer_estimate}) and the virtual control input in Eq. (\ref{eq:v_out}) are described as follows:
\begin{equation}
    {\mathbf{e_\textbf{out}}} = \left[\,e_z\:\:\:\:\dot{e}_z\:\:\:\:\zeta_z\,\right]^T\,,
\end{equation}
\begin{equation}\label{eq:outer_edyn}
    \dot{\mathbf{e}}_\textbf{out} = \begin{bmatrix}
        0 & 1 & 0\\
        -k_z & -k_{\dot{z}} & -k_i\\
        1 & 0 & 0
    \end{bmatrix}\,\mathbf{e_\textbf{out}} +
    \begin{bmatrix}
        0\\
        \frac{\prescript{i}{}{\tilde{f}_{a_z}}}{m}\\
        0
    \end{bmatrix}\,,
\end{equation}
where $\prescript{i}{}{\tilde{f}_{a_z}} = \prescript{i}{}{f_{a_z}} - \prescript{i}{}{\hat{f}_{a_z}}$ is the horizontal aerodynamic force estimation error. Similarly to the inner-loop, the outer closed-loop error dynamics can be seen as a stable linear system with an input disturbance given by the estimation error.
\begin{figure*}
\centering
\includegraphics{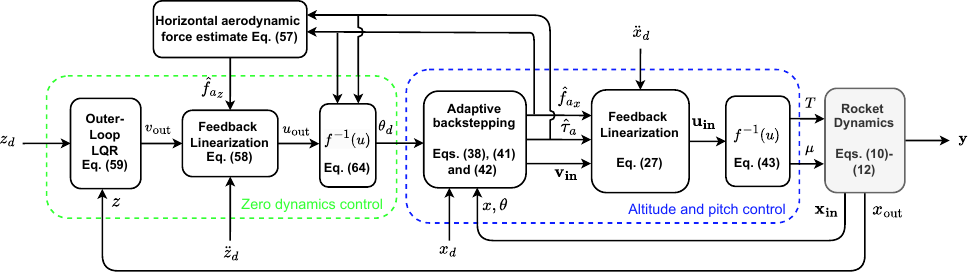}
\caption{\bf{Diagram of the proposed inner-outer control architecture.}}
\label{fig:arch_dia}
\end{figure*}\\

\begin{theorem}\label{T:outer}
    Under the conditions stated in Theorem \ref{T:inner}, and considering the aerodynamic moment modeling assumption, the origin of the outer closed-loop error dynamics given by Eq. (\ref{eq:outer_edyn}) is regionally asymptotically stable for $|\theta|<\pi/2$. 
\end{theorem}
\begin{proof}
    From Proposition \ref{prop:out}, the input-output feedback linearization is well-defined for all $\mathbf{x_\textbf{out}}$. Thus, considering a perfect estimation of the horizontal aerodynamic force, $\prescript{i}{}{f_{a_z}}$, the application of the input-output linearization described in (\ref{eq:outer_lin_hat}) with virtual control given by Eq. (\ref{eq:v_out}) results in a Hurwitz characteristic polinomial of the double integrator in closed-loop. Hence, the linear state variables $\boldsymbol{\xi}_\textbf{out}$ converge exponentially to the desired horizontal position trajectory for any initial state $\boldsymbol{\xi}_\textbf{out}(0)$. Moreover, from the Hurwitz condition, the closed-loop system (\ref{eq:outer_edyn}) is input-to-state stable, meaning that a bounded input to the system (\ref{eq:outer_edyn}), i.e., a bounded estimation error of the horizontal aerodynamic force, results in bounded state trajectories, and convergence of $\prescript{i}{}{\tilde{f}_{a_z}}$ to zero implies the outer-loop tracking error converging to zero. From Theorem \ref{T:inner}, and under the conditions therein stated, the inner-loop estimation errors $\prescript{i}{}{\tilde{f}_{a_x}}$ and $\tilde{\tau}_a$ converge to zero, which implies, under the aerodynamic moment modeling assumption, the convergence of $\prescript{i}{}{\tilde{f}_{a_z}}$ to zero. Furthermore, the pitch angle $\theta$ is asymptotically stabilized by the inner-loop for the region defined by $|\theta|<\pi/2$. Given that $\boldsymbol{\xi}_\textbf{out} = \mathbf{x_\textbf{out}}$, the horizontal motion tracking dynamics (\ref{eq:outer_edyn}) are asymptotically stable for $|\theta|<\pi/2$.
\end{proof}

From Theorem \ref{T:outer}, we have that the outer-loop dynamics, obtained through the \textit{zero dynamics} of the inner-stabilized system, are now stable and able to track generic horizontal position references $z_d$. Figure \ref{fig:arch_dia} schematizes the resulting control system. While the inertial position references are externally defined according to the mission scenario, the pitch angle references $\theta_d$ must be determined by the outer-loop control in order to track the desired horizontal position. Hence, the pitch angle references must be extracted from the outer-loop input in Eq. (\ref{eq:outerinput}). Rewriting the outer-loop input as 
\begin{equation}\label{eq:newouterinput}
     u_\text{out}= \hat{a}\,\tan{\theta_d} +  \frac{\hat{b}}{\cos{\theta_d}}\,,
\end{equation}
the reference pitch angle $\theta_d$ can be obtained:
\begin{equation}\label{eq:pitch_ref}
    \theta_d = \cos^{-1}\left(\frac{\hat{b}}{\sqrt{\hat{a}^2+{u_\text{out}}^2}}\right) + \arctan_2\left(\hat{a},-u_\text{out}\right)\,.
\end{equation}

In order to study the stability of the complete system, the interconnection of the inner and outer loops must be analyzed. To do so, the rewritten outer-loop input in Eq. (\ref{eq:newouterinput}) shall be included in the outer error dynamics system detailed in Eq. (\ref{eq:outer_xdot}), yielding
\begin{equation}
\renewcommand{\arraystretch}{1.5}
    \dot{{\mathbf{x}}}_\textbf{out} = \begin{bmatrix}
         \dot{e}_z\\
         \frac{\prescript{i}{}{f_{a_z}}}{m}-\ddot{z_d}
    \end{bmatrix}+
    \begin{bmatrix}
        0\\
        1
    \end{bmatrix}\,u_\text{out}+
    \begin{bmatrix}
        0\\
        \delta
    \end{bmatrix}\,,
\end{equation}
in which $\delta$ corresponds to the impact of the pitch angle tracking error on the outer-loop dynamics and is given by
\begin{equation}\label{eq:delta}
     \delta =  \left(a\,\tan{\theta} +  \frac{b}{\cos{\theta}}\right)-u_\text{out}\,.
\end{equation}
Moreover, the outer closed-loop error dynamics (\ref{eq:outer_edyn}) reshape into:
\begin{equation}\label{eq:outer_edyn_delta}
\renewcommand{\arraystretch}{1.5}
    \dot{\mathbf{e}}_\textbf{out} = \begin{bmatrix}
        0 & 1 & 0\\
        -k_z & -k_{\dot{z}} & -k_i\\
        1 & 0 & 0
    \end{bmatrix}\,\mathbf{e_\textbf{out}} +
    \begin{bmatrix}
        0\\
        \frac{\prescript{i}{}{\tilde{f}_{a_z}}}{m}\\
        0
    \end{bmatrix} +
    \begin{bmatrix}
        0\\
        \delta\\
        0
    \end{bmatrix}.
\end{equation}
\begin{theorem} \label{T:cascade}
    Consider the closed-loop system composed by the nonlinear system described by Eq. (\ref{eq:sys_xdot}) and the input-output linearizing controller given by (\ref{eq:inner_lin}) and (\ref{eq:virtual_input_law}). Let the remaining dynamics, corresponding to the horizontal movement dynamics (\ref{eq:outer_xdot}), be input-output linearized and controlled through (\ref{eq:outer_lin_hat}) and (\ref{eq:v_out}). The resulting closed-loop tracking system is regionally asymptotically stable.
\end{theorem}
\begin{proof}
    From Theorem \ref{T:inner}, the application of the input-output linearizing controller defined by (\ref{eq:inner_lin}) and (\ref{eq:virtual_input_law}) to the system (\ref{eq:sys_xdot}) while considering the output vector (\ref{eq:inner_y}) results in a closed-loop system, concerning the altitude and attitude dynamics, with error dynamics of the type (\ref{eq:inner_edyn}) that ensures asymptotic convergence of the state vector $\mathbf{x_{in}}$ to the desired values for $|\theta|<\pi/2$. Considering that the nonlinear system (\ref{eq:sys_xdot}) is composed by 6 state variables, and the sum of the entries of the relative degree vector $\mathbf{r_{in}}$ for the output (\ref{eq:inner_y}) is equal
to 4, there is an unobservable subsystem, characterized by the dynamics associated with the remaining 2 state variables, that corresponds to the horizontal movement dynamics expressed in (\ref{eq:outer_xdot}) and (\ref{eq:outer_y}). From Theorem \ref{T:outer}, the closed-loop horizontal dynamics resulting from the application of the input-output linearizing controller defined by (\ref{eq:outer_lin_hat}) and (\ref{eq:v_out}) has a Hurwitz characteristic polynomial. Hence, convergence of the input $\delta$ (\ref{eq:delta}) to zero in the tracking system (\ref{eq:outer_edyn_delta}) implies the convergence of the tracking error to zero. Given the stabilizing inner-loop adaptive tracking controller, $\theta$ asymptotically converges to $\theta_d$ and the inner-loop aerodynamic estimates converge to the true values, which  we have that $\delta$ converges to zero. Hence, the tracking system (\ref{eq:outer_edyn_delta}) is asymptoticaly stable for $|\theta|<\pi/2$. Thus, the complete cascaded tracking system resulting from the application of the inner input-output linearizing controller given by (\ref{eq:inner_lin}) and (\ref{eq:virtual_input_law}), and the outer input-output linearizing controller defined by (\ref{eq:outer_lin_hat}) and (\ref{eq:v_out}) is asymptotically stable for $|\theta|<\pi/2$.
\end{proof}

From Theorem \ref{T:cascade}, it has been established that the proposed inner-outer control solution is stable in the region of interest and able to track sufficiently smooth arbitrary trajectories inside the pitch plane, as long as the physical limitations of the vehicle are respected.
\section{Implementation in Simulation} \label{sec:implement}

To test the proposed architecture, a realistic simulation environment, composed by the pitch plane non-linear dynamics model, the control system, and the environmental properties, was implemented in Matlab\&Simulink\textsuperscript{\circledR{}}. Additionally, a reference vehicle and a set of mission scenarios had to be selected. In this section, the simulation environment, reference vehicle, and mission scenarios, as well as the implementation details of the control architecture, are presented.

\subsection{Reference Vehicle}

The reference vehicle preliminary design was performed considering the typical characteristics and performance parameters of suborbital launch vehicles currently under operation so as to test the simulation tool with realistic design values. More specifically, liquid engine vehicles with thrust vector control capabilities were considered, with an available total impulse which allows reaching space according to the Kármán line definition (approximately $100\,$km altitude). Figure \ref{fig:rocket_scheme} shows a schematic of the vehicle with some important dimensions, while Table \ref{tab:rocket_char} sums up the main properties of the reference vehicle.
\begin{figure*}
\centering
\includegraphics{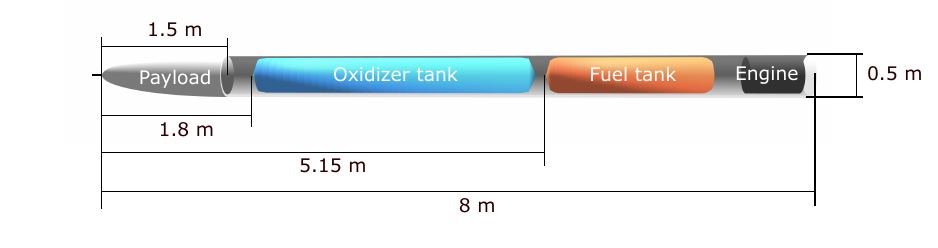}
\caption{\bf{Scheme of the vehicle.}}
\label{fig:rocket_scheme}
\end{figure*}
\begin{table}
\renewcommand{\arraystretch}{1.3}
\caption{\bf Main vehicle characteristics}
\label{tab:rocket_char}
\centering
\begin{tabular}{|c|c|}
\hline
\bfseries Parameter & \bfseries Value \\
\hline\hline
Length      &  8\,m \\
Max diameter & 50\,cm \\
Lift-off mass & 1250\,kg\\
Dry mass & 600\,kg\\
Dry inertia ($j_y$) & 3100\,kg.m$^2$\\
Dry $x_{cm}$ & 4.18\,m\\
\hline
\end{tabular}
\end{table}
An axially symmetric geometry was selected, composed by a cylindrical fuselage with an elliptical nose cone. Given that stability is to be provided by the thrust vector control system, no aerodynamic fins were added for passive stability. Using the selected geometry, the aerodynamic coefficients and center of pressure location were obtained through Computational Fluid Dynamics (CFD) simulations and stored as functions of the angle of attack and Mach number to be interpolated during simulation. Regarding the propulsion system, the well-established liquid oxygen and refined querosene oxidizer/fuel pair (LOX/RP1) was considered. Propulsion system parameters are collected in Table \ref{tab:prop_char}.

\begin{table}
\renewcommand{\arraystretch}{1.3}
\caption{\bf Propulsion system parameters}
\label{tab:prop_char}
\centering
\begin{tabular}{|c|c|}
\hline
\bfseries Parameter & \bfseries Value \\
\hline\hline
 Total propellant mass & 650\,kg   \\
$I_{\text{sp}_0}$               & 300\,s  \\
Oxidizer/fuel mass ratio & 2.5\\
LOX density & 1141\,kg.m$^{-2}$\\
RP1 density & 800\,kg.m$^{-2}$\\
Oxidizer tank volume & 0.41\,m$^{3}$\\
Fuel tank volume & 0.23\,m$^{3}$\\
\hline
\end{tabular}
\end{table}

\subsection{Mission Scenarios}

Two distinct mission scenarios were defined having in mind potential applications for actively controlled sounding rockets. Both scenarios refer to the propulsive stage of the flight, i.e., from lift-off until engine shut-off, during which thrust vectoring control can be used to actively steer the vehicle. It is assumed that during a posterior coasting phase, occurring in low density atmospheric regions, the vehicle may be stabilized by an additional reaction control system if necessary. In order to ensure sufficiently smooth reference trajectories, the desired time evolution of the inertial position is prescribed in terms of the inertial accelerations, $\ddot{x}_d$ and $\ddot{z}_d$, which are then integrated to retrieve the desired position and velocity.

\subsubsection{Mission Scenario I}
The first scenario consists of a strictly vertical trajectory, i.e, the goal of the trajectory tracking controller is to prevent horizontal motion in the presence of external disturbances. 

The desired inertial acceleration for this scenario is given by
\begin{subequations}\label{eq:scenI}
\begin{equation}\label{eq:scenIvert}
  \ddot{x}_d(t) = 
  \begin{cases}
  10\,,\hspace{82pt} 0<t\leq30 \\[5pt]
 2.5\cos\left(\frac{2\pi}{30}t\right)+7.5\,,\hspace{10pt} 30<t\leq60 \\[5pt]
  10\,,\hspace{79pt} 60<t<t_{f} 
  \end{cases}\,,
 \end{equation}
 \begin{equation}
  \ddot{z}_d(t) = 0\,.
\end{equation}
\end{subequations}
Looking at (\ref{eq:scenI}), the desired vertical acceleration is set at a constant value of 10\,m.s$^{-2}$ until burnout, at $t=t_f$, apart from a period during which it follows a sinusoidal curve. The sinusoidal curve translates into a momentary reduction in acceleration to minimize the loads on the vehicle during the maximum dynamic pressure region, also known as \textit{Max Q}. Desired position and velocity are obtained by integration considering zero initial conditions, leading to an altitude at burnout of 48.8\,km and a vertical velocity of 955\,m.s$^{-1}$.

\subsubsection{Mission Scenario II}
The second mission scenario has the same vertical acceleration profile (Eq. \ref{eq:scenIvert}) while imposing horizontal motion. Horizontal motion is imposed with the objective of initially deviating the vehicle from the launch site and then recovering a vertical trajectory before burnout. This trajectory requirement may be useful to increase safety, to perform remote sensing either in an atmospheric or ground region away from the launch site, and/or to adjust the potential landing zone for both guided and unguided recovery.    
The desired horizontal acceleration for this scenario is given by
\begin{equation}\label{eq:scenII}
  \ddot{z}_d(t) = 
  \begin{cases}
  0\,,\hspace{82pt} 0<t\leq20 \\[5pt]
 2.5\left[\cos\left(\frac{2\pi}{40}t\right)+1\right]\,,\hspace{10pt} 20<t\leq60 \\[5pt]
  -2.5\left[\cos\left(\frac{2\pi}{40}t\right)+1\right]\,,\hspace{5pt} 60<t<t_{f} 
  \end{cases}\,,
\end{equation}
Once again, inertial velocity and position are obtained by integration considering zero initial conditions, leading to a maximum downrange distance of 4\,km.

\subsection{Simulation Model}

The simulation model is divided into two major components: the trajectory tracking control architecture, which corresponds to the computational implementation of the diagram in Fig. \ref{fig:arch_dia}, and the vehicle physical model. The physical model computes the evolution of the flight variables by integrating the set of non-linear pitch plane equations of motion presented in (\ref{eq:position})-(\ref{eq:attitudekin}), given the inputs determined by the control architecture. To do so, it needs to generate the time-varying parameters of the vehicle, such as mass, inertia, and aerodynamic coefficients, as well as the environmental conditions, namely the atmospheric pressure, air density, wind, speed of sound, and gravitational acceleration. The evolution of the MCI properties of the vehicle is determined by the propellant consumption given in terms of mass flow rate (Eq. \ref{eq:massflow}), considering evenly distributed mass with respect to the longitudinal axis and neglecting propellant sloshing. The aerodynamic coefficients and center of pressure location are stored in two-dimensional look-up tables, available as Simulink blocks, which perform interpolation to retrieve the instantaneous values depending on the current angle of attack and Mach number. As for the atmospheric conditions, the 1976 U.S standard atmosphere model was implemented, which describes the evolution of temperature and pressure with altitude using average annual values, from which density and speed of sound are derived. Wind is introduced through the summation of the average horizontal wind components from the U.S Naval Research Laboratory horizontal wind model with a stochastic component (wind gusts) added from the Dryden model, both available as Simulink blocks. Finally, the gravitational acceleration is computed according to Eq. (\ref{eq:gravity}).

\subsection{Architecture Implementation Details}

The computational implementation of the proposed inner-outer trajectory tracking control architecture follows the scheme in Fig. \ref{fig:arch_dia}. Nevertheless, some implementation details will be addressed to ensure the reproducibility of the simulation results. Starting from the inner-loop, Table \ref{tab:inner_gains} presents the selected control and adaptation gains.
\begin{table}[h!]
\renewcommand{\arraystretch}{1.3}
\caption{\bf Inner-loop gains}
\label{tab:inner_gains}
\centering
\begin{tabular}{|c|c|c|c|c|c|}
\hline
 $k_{1_x}$ & $k_{2_x}$ & $k_{1_\theta}$ & $k_{2_\theta}$ & $\gamma_x$ &  $\gamma_\theta$\\
\hline\hline
2.5 & 4.5 & 12 & 10 & 5 & 5   \\
\hline
\end{tabular}
\end{table}
We note that the derived stability of the inner-loop implies positive control and adaptation gains. Considering that constraint, the gains were tuned to ensure a sufficiently fast inner-loop response while avoiding excessive input values. Regarding the adaptation gains, the values were set so that the estimates are able to track the time-varying vertical aerodynamic force and moment with small errors, however, due to the stochastic and noisy nature of the wind gusts, the gain for the aerodynamic moment estimate was selected so as to filter the high frequency disturbances and avoid highly oscillatory response. Moreover, the magnitude of the control and adaptation gains determines the bandwidth of the actuation signals, which has to respect the limited bandwidth of the on-board actuators, representing an additional constraint on gain selection. 

Regarding the outer-loop, the chosen LQR tuning matrices are the following:
\begin{equation}
\renewcommand{\arraystretch}{1.3}
{\bf Q} =\begin{bmatrix}
    5 & 0 & 0\\
    0 & 1 & 0\\
    0 & 0 & 1
\end{bmatrix}\,,\hspace{10pt}
{\bf R} = 60\,,
\end{equation}
which lead to the gains present in Table \ref{tab:outer_gains}.
\begin{table}[h!]
\renewcommand{\arraystretch}{1.3}
\caption{\bf Outer-loop gains}
\label{tab:outer_gains}
\centering
\begin{tabular}{|c|c|c|}
\hline
 $k_{z}$ & $k_{\dot{z}}$ & $k_i$ \\
\hline\hline
0.61 & 1.11 & 0.13  \\
\hline
\end{tabular}
\end{table}
The LQR tuning matrices were selected to ensure sufficient time-scale separation between the outer and inner loops, while limiting the horizontal position tracking error. Additional caution was taken when differentiating the pitch reference command generated by the outer-loop to obtain the angular velocity and acceleration commands by adding first-order low-pass filters, given that the output pitch reference may carry noise from the real-time aerodynamic estimates.

\section{Simulation Results}\label{sec:sim_reuslts}

Using the reference vehicle parameters inside the simulation model, the two previously introduced mission scenarios were simulated to verify the proposed control architecture and derive its performance. Firstly, the two scenarios were simulated for the nominal case, i.e, assuming complete knowledge of the model parameters and for a given wind condition. Then, a Monte Carlo robustness analysis was performed where uncertainty was added to model parameters and the noise seed associated with the wind gusts was changed between runs. The performance of the proposed control system is evaluated in the absence of sensor noise and assuming full-state knowledge, given that the state variables in case can be estimated by an inertial navigation system as the one proposed by the authors in \cite{dossantos}. The nominal wind profile is depicted in Fig. \ref{fig:wind}. The impact of the wind gusts generated through the Dryden model is visible up to approximately 20\,km, after which the average wind dominates.  
\begin{figure}[htpb]\label{fig:wind}
\centering
\includegraphics{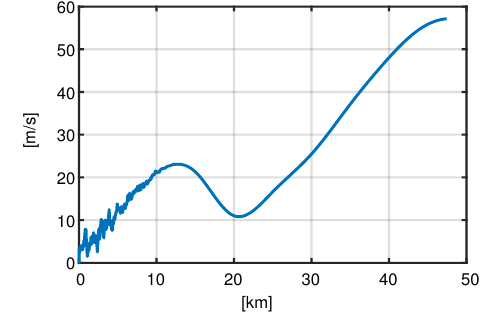}\\
\caption{\textbf{Wind profile.}}
\end{figure}

\subsection{Mission Scenario I}

Starting with the first mission scenario (vertical trajectory), Figures \ref{fig:vert_traj} to \ref{fig:vert_actuation} display the most relevant simulated flight, control, and estimation variables, while Tables \ref{tab:track_one} and \ref{tab:est_one} present control and estimation performance metrics, respectively. Looking at Fig. \ref{fig:vert_traj}, it is seen that the vehicle is able to track the desired vertical trajectory with limited tracking error, which is also backed by the root mean square error (RSME) values in Table \ref{tab:track_one}.  
\begin{table}[h]
\renewcommand{\arraystretch}{1.3}
\caption{\bf Tracking performance and control effort for mission scenario I}
\label{tab:track_one}
\centering
\begin{tabular}{|c|c|c|c|c|}
\hline
\bfseries y & \bfseries RMSE & &\bfseries u & \bfseries RMS\\
\hline\hline
$x$   &  0.008\,m && $T$ & 18.52\,kN\\
$\theta$ & 0.06\,$^{\circ}$ &&$\mu$ & 0.51\,$^{\circ}$ \\
$z$ & 0.11\,m && & \\
\hline
\end{tabular}
\end{table}
The impact of the wind on the tracking performance is clearly visible at the initial segment of the flight, where wind gusts are present and dynamic pressure is higher. Furthermore, the horizontal tracking performance is considerably lower than the vertical one, which is also attributed to the horizontal direction of the wind disturbance. From Fig. \ref{fig:vert_vel}, we see that the desired pitch angle and inertial velocity are also correctly tracked by the vehicle. Once again, the impact of the wind gusts is seen on the horizontal velocity, as well as on the pitch reference generated by the outer-loop. In fact, horizontal position control is enabled by the adaptation of the pitch reference to the external wind disturbance, which is then correctly tracked by the inner-loop controller. 

Figure \ref{fig:vert_estimation} displays the estimation of the inertial aerodynamic forces and moment. The aerodynamic estimates are able to follow the true values, which was to be expected given the good trajectory tracking performance,. The wind disturbance causes higher impact on the aerodynamic moment and on the horizontal aerodynamic force, resulting in worse estimation performances when compared to the vertical force estimation. This fact is also verified by the estimation RMSE and relative RMSE (\%RMSE) presented in Table \ref{tab:est_one}. Moreover, higher time delay for the moment and horizontal force estimation is verified, which is the result of setting the correspondent gains at a value that filters the impact of the wind gusts at the cost of an increased delay.
\begin{table}
\renewcommand{\arraystretch}{1.4}
\caption{\bf Aerodynamic estimation performance for mission scenario I}
\label{tab:est_one}
\centering
\begin{tabular}{|c|c|c|}
\hline
\bfseries Parameter & \bfseries RMSE & \bfseries \%RMSE\\
\hline\hline
$\prescript{i}{}{f_{a_x}}$   &  92.6\,N & 7.0\,\%\\
$\tau_a$ & 124.0\,N.m & 26.0\,\% \\
$\prescript{i}{}{f_{a_z}}$  & 37.2\,N & 21.1\,\% \\
\hline
\end{tabular}
\end{table}
By comparing the horizontal aerodynamic force estimate and the generated pitch reference, it is possible to conclude a clear similarity in shape, which indicates once more that, as seen in Eq. (\ref{eq:pitch_ref}), the pitch reference generated by outer-loop controller strongly depends on the horizontal aerodynamic load, which is a mark of the adaptive nature of the proposed architecture. It is also relevant to note that the proposed adaptation laws are able to track highly time-varying aerodynamic loads, even though the design condition assumed constant behavior. This means that the simulation results validate the theoretical assumption that sufficiently high adaptation gains allow to track time-varying signals. 

Finally, Figure \ref{fig:vert_actuation} displays the actuation signals, i.e, the thrust magnitude ($T$) and the thrust vector deflection angle ($\mu$), over time. Additionally, the mass flow rate ($|\dot{m}|$) is also depicted. The thrust magnitude computed by the inner-loop controller follows a decreasing trend given by the constant segments of the desired vertical acceleration and the mass reduction over time. The impact of the requested vertical acceleration reduction is also visible between 30 and 60 seconds after lift-off. As for the thrust vector deflection angle, it exhibits a similar trend to both the pitch reference and horizontal aerodynamic force, translating the impacting of the average horizontal wind summed with the stochastic gusts. The root mean square (RMS) of the actuation signals is collected in Table \ref{tab:track_one}.

\begin{figure*}
\centering
\includegraphics{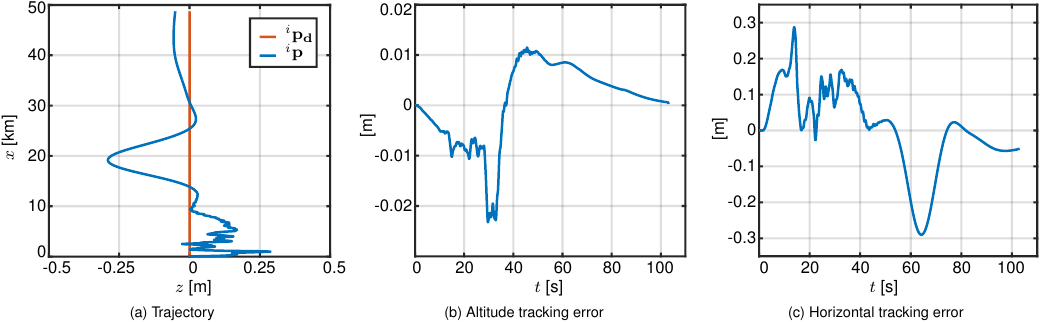}
\caption{\bf{Trajectory tracking for mission scenario I.}}
\label{fig:vert_traj}
\end{figure*}
\begin{figure*}
\centering
\includegraphics{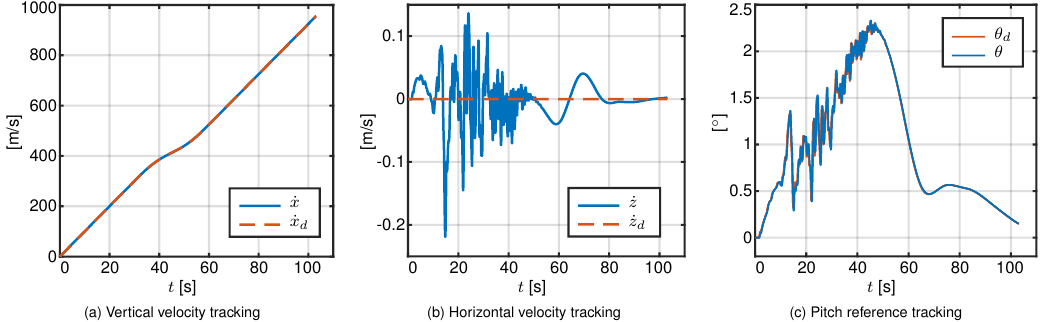}
\caption{\bf{Velocity and pitch reference tracking for mission scenario I.}}
\label{fig:vert_vel}
\end{figure*}
\begin{figure*}
\centering
\includegraphics{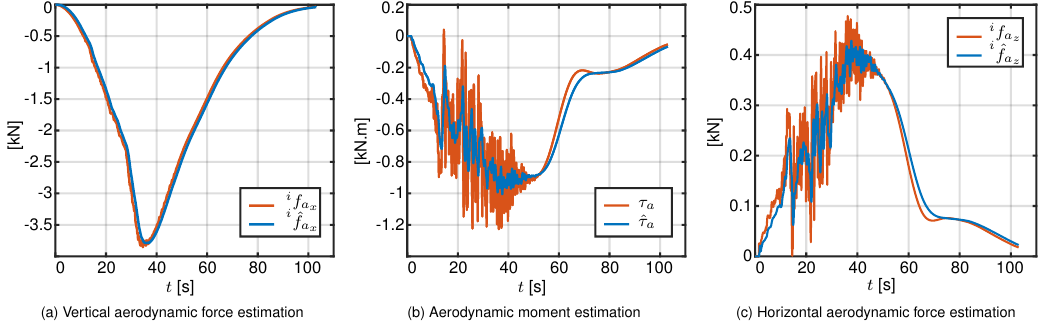}
\caption{\bf{Aerodynamic forces and moment estimation for mission scenario I.}}
\label{fig:vert_estimation}
\end{figure*}
\begin{figure*}
\centering
\includegraphics{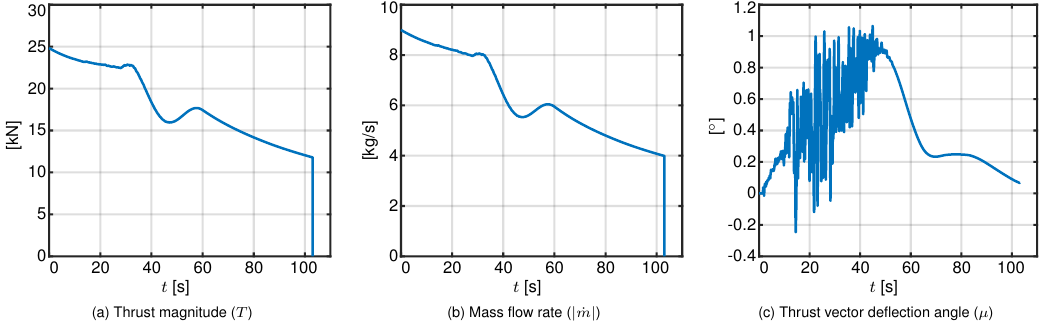}
\caption{\bf{Actuation for mission scenario I.}}
\label{fig:vert_actuation}
\end{figure*}
\begin{figure*}
\centering
\includegraphics{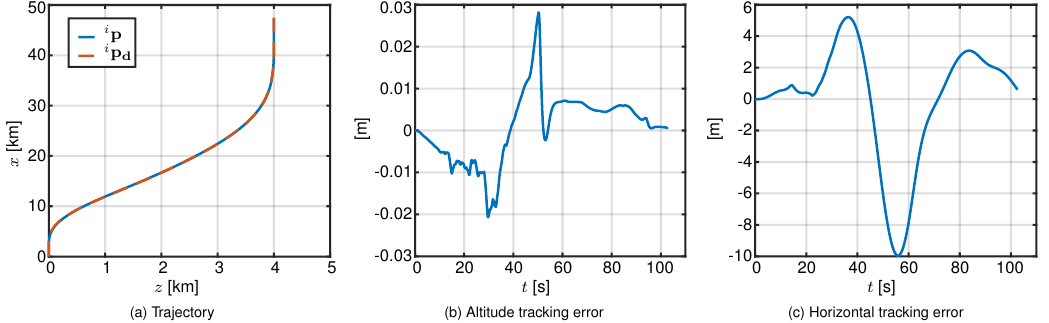}
\caption{\bf{Trajectory tracking for mission scenario II.}}
\label{fig:lat_traj}
\end{figure*}
\begin{figure*}
\centering
\includegraphics{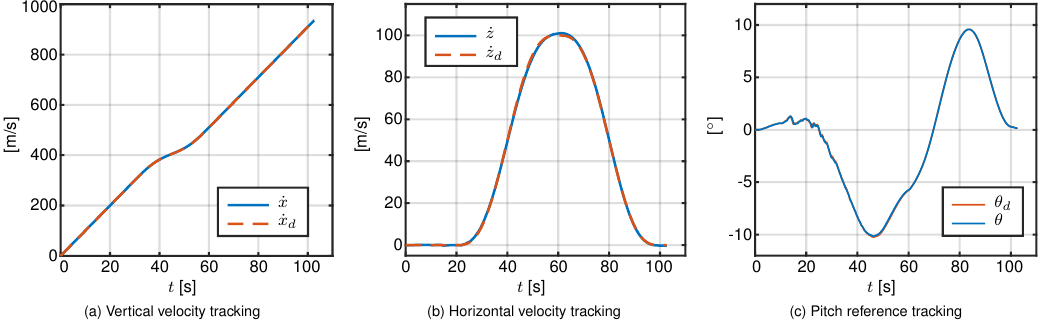}
\caption{\bf{Velocity and pitch reference tracking for mission scenario II.}}
\label{fig:lat_vel}
\end{figure*}
\begin{figure*}
\centering
\includegraphics{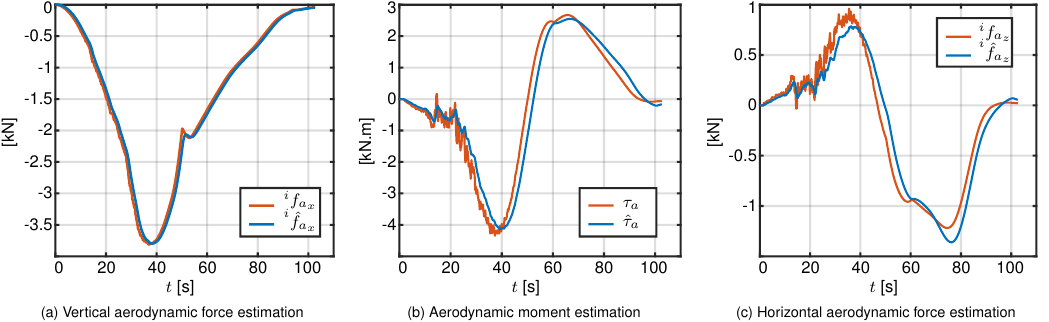}
\caption{\bf{Aerodynamic forces and moment estimation for mission scenario II.}}
\label{fig:lat_estimation}
\end{figure*}
\begin{figure*}
\centering
\includegraphics{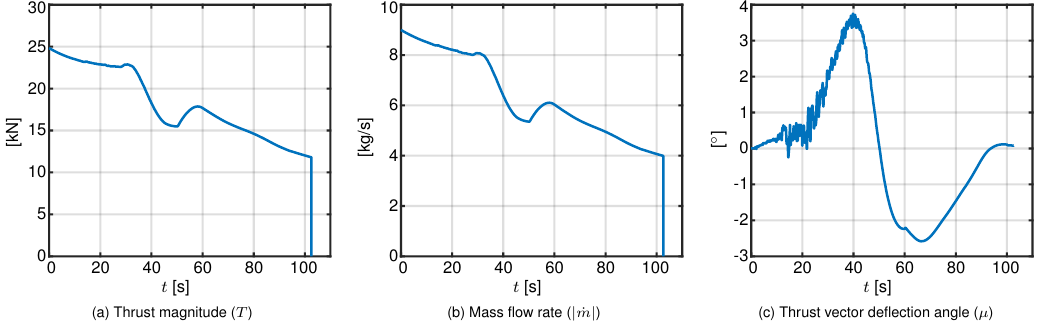}
\caption{\bf{Actuation for mission scenario II.}}
\label{fig:lat_actuation}
\end{figure*}

\subsection{Mission Scenario II}

Moving on to the second mission scenario, Figures \ref{fig:lat_traj} to \ref{fig:lat_actuation} display the most relevant simulated flight, control, and estimation variables, while Tables \ref{tab:track_two} and \ref{tab:est_two} present control and estimation performance metrics, respectively. In Fig. \ref{fig:lat_traj}, the trajectory and the corresponding position tracking errors are displayed. Using the same parameters for the control architecture, i.e., without re-tuning the gains, the rocket is able to track a more demanding trajectory in terms of horizontal movement. Table \ref{tab:track_two} presents the tracking performance and control effort metrics for this mission scenario.
\begin{table}[h]
\renewcommand{\arraystretch}{1.3}
\caption{\bf Tracking performance and control effort for mission scenario II}
\label{tab:track_two}
\centering
\begin{tabular}{|c|c|c|c|c|}
\hline
\bfseries y & \bfseries RMSE & &\bfseries u & \bfseries RMS\\
\hline\hline
$x$   &  0.008\,m && $T$ & 18.59\,kN\\
$\theta$ & 0.07\,$^{\circ}$ &&$\mu$ & 1.76\,$^{\circ}$ \\
$z$ & 0.96\,m && & \\
\hline
\end{tabular}
\end{table}
By comparison with the values obtained for the first scenario, we see that the RSME associated with horizontal position tracking is now superior, as it would be expected due to the non-zero horizontal acceleration reference provided to the outer-loop. Looking at Fig. \ref{fig:lat_vel}, it is observed that a maximum horizontal velocity of 100\,m/s is reached, coinciding with the largest error in horizontal velocity tracking. The pitch reference generated by the outer-loop controller is now substantially more demanding and reflects the desired horizontal motion, i.e, the vehicle turns downrange and then turns back to recover a vertical trajectory. Nevertheless, the inner-loop controller is still able to track the required pitch angle with similar performance.

Regarding the aerodynamic estimation performance, Fig. \ref{fig:lat_estimation} displays the estimates over time, while Table \ref{tab:est_two} presents the estimation performance metrics. Through visual inspection of the plots, we see that the aerodynamic estimation for this scenario was also successful. As expected, higher horizontal loads are experienced during this trajectory, which reflects into worse relative estimation performance for the aerodynamic moment and horizontal aerodynamic force when compared to the first scenario.
\begin{table}[h]
\renewcommand{\arraystretch}{1.4}
\caption{\bf Aerodynamic estimation performance for mission scenario II}
\label{tab:est_two}
\centering
\begin{tabular}{|c|c|c|}
\hline
\bfseries Parameter & \bfseries RMSE & \bfseries \%RMSE\\
\hline\hline
$\prescript{i}{}{f_{a_x}}$   &  96.1\,N & 6.8\,\%\\
$\tau_a$ & 457.9\,N.m & 29.7\,\% \\
$\prescript{i}{}{f_{a_z}}$  & 147.6\,N & 28.4\,\% \\
\hline
\end{tabular}
\end{table}
The worse estimation performance propagates into the trajectory tracking performance, further justifying the increase in the horizontal position tracking error. Nevertheless, stability is still ensured and the tracking performance is deemed acceptable, specially when considering the considerable magnitude of the velocity vector.

Given that the vertical position reference is the same as before, the thrust magnitude computed by the inner-loop is very similar, as seen in Fig. \ref{fig:lat_actuation} and attested by the RMS of the signal shown in Table \ref{tab:track_two}. On the other hand, the thrust vector deflection angle has now a higher RMS value as a result of the non-zero pitch reference required to track the desired downrange motion. The wind gusts impact is still very noticeable on the deflection angle, causing oscillatory behavior, although now with a lower relative magnitude.

\subsection{Robustness Analysis}

The robustness of the proposed trajectory tracking controller was evaluated resorting to Monte Carlo (MC) techniques. For each mission scenario, a total of 100 MC runs were conducted in which several rocket parameters were randomly changed in between consecutive runs, while feeding the nominal values to the controller. The wind gusts generation seeds were also randomly generated. This analysis gains more relevance when considering the indirect estimation of the horizontal aerodynamic force, which relies on the knowledge of the static stability margin, and the use of several MCI properties on the feedback linearization laws. Table \ref{tab:rob_param} shows the rocket parameters subject to random variability, which was obtained through the multiplication of said parameters by unitary mean Gaussian distributions, as well as the associated 3$\sigma$ values.
\begin{table}[h]
\renewcommand{\arraystretch}{1.3}
\caption{\bf Selected parameters for robustness analysis}
\label{tab:rob_param}
\centering
\begin{tabular}{|c|c|}
\hline
\bfseries Parameter & \bfseries 3\,$\sigma$ \\
\hline\hline
$m$      &  0.05 \\
$j_y$      &  0.1 \\
$x_{cm}$ & 0.1 \\
 $x_{cp}$ & 0.2\\
$C_D$ & 0.2\\
$C_L$ & 0.2\\
\hline
\end{tabular}
\end{table}
Note that by introducing uncertainty on both the center of mass and center of pressure it propagates into the control moment arm, $l$, and static stability margin, $SM$.

After completing all MC runs, it was verified that the proposed architecture was able to stabilize the vehicle and track the desired trajectory for all the runs and for both mission scenarios. This indicates that the proposed solution is robust to both external disturbances and model uncertainty. Nonetheless, this conclusion is limited to the assumed level of uncertainty. Tables \ref{tab:rob_one} and \ref{tab:rob_est_one} respectively present the robustness analysis results for the trajectory tracking and aerodynamic estimation performance for the first mission scenario. The results are presented in terms of the mean and standard deviation (STD) of each performance metric.
\begin{table}
\renewcommand{\arraystretch}{1.3}
\caption{\bf MC results for tracking performance and control effort - Mission scenario I}
\label{tab:rob_one}
\centering
\begin{tabular}{|c|c|c|}
\hline
\bfseries y & \bfseries Average RMSE & \bfseries STD \\
\hline\hline
$x$   &  0.008\,m &  0.0006\,m\\
$\theta$ & 0.06\,$^{\circ}$ & 0.03\,$^{\circ}$\\
$z$ & 0.12\,m & 0.009\,m\\
\hline
\hline
\bfseries u & \bfseries Average RMS & \bfseries STD\\
\hline
\hline
 $T$ & 18.37\,kN & 0.67\,kN \\
 $\mu$ & 0.52\,$^{\circ}$& 0.06\,$^{\circ}$ \\
\hline
\end{tabular}
\end{table}
\begin{table}
\renewcommand{\arraystretch}{1.4}
\caption{\bf MC results for aerodynamic estimation performance - Mission scenario I}
\label{tab:rob_est_one}
\centering
\begin{tabular}{|c|c|c|}
\hline
\bfseries Parameter & \bfseries Average \,\%RMSE & \bfseries STD\\
\hline\hline
$\prescript{i}{}{f_{a_x}}$   & 13.8\,\% & 8.0\,\% \\
$\tau_a$ & 21.0\,\% & 2.4\,\% \\
$\prescript{i}{}{f_{a_z}}$  &  18.2\,\% & 6.2\,\% \\
\hline
\end{tabular}
\end{table}
It is noteworthy that, overall, the performance has not deteriorated in the presence of uncertainty. The largest decrease in  performance is identified for the vertical aerodynamic force estimation, which is attributed to its dependency on the mass of the vehicle, a parameter which was subject to uncertainty.

Likewise, the MC results for the second mission scenario are collected in Tables \ref{tab:rob_two} and \ref{tab:rob_est_two}.
\begin{table}[t]
\renewcommand{\arraystretch}{1.3}
\caption{\bf MC results for tracking performance and control effort - Mission scenario II}
\label{tab:rob_two}
\centering
\begin{tabular}{|c|c|c|}
\hline
\bfseries y & \bfseries Average RMSE & \bfseries STD \\
\hline\hline
$x$   &  0.008\,m &  0.0007\,m\\
$\theta$ & 0.07\,$^{\circ}$ & 0.007\,$^{\circ}$\\
$z$ & 0.97\,m & 0.04\,m\\
\hline
\hline
\bfseries u & \bfseries Average RMS & \bfseries STD\\
\hline
\hline
 $T$ & 18.62\,kN & 0.69\,kN \\
 $\mu$ & 1.80\,$^{\circ}$& 0.52\,$^{\circ}$ \\
\hline
\end{tabular}
\end{table}
\begin{table}[t]
\renewcommand{\arraystretch}{1.4}
\caption{\bf MC results for aerodynamic estimation performance - Mission scenario II}
\label{tab:rob_est_two}
\centering
\begin{tabular}{|c|c|c|}
\hline
\bfseries Parameter & \bfseries Average \,\%RMSE & \bfseries STD\\
\hline\hline
$\prescript{i}{}{f_{a_x}}$   & 13.2\,\% & 7.4\,\% \\
$\tau_a$ & 23.1\,\% &0.6\,\% \\
$\prescript{i}{}{f_{a_z}}$  &  27.8\,\% & 6.7\,\% \\
\hline
\end{tabular}
\end{table}
Also for this scenario the largest performance decrease happens for the vertical aerodynamic force estimation, with the remaining metrics showing no substantial shift.

\section{Conclusions}\label{sec:conc}

With the completion of this work, a pitch plane trajectory tracking controller for underactuated sounding rockets has been successfully derived and tested in simulation. Preliminary results indicate that the methodology adopted by the authors is able to track arbitrary, sufficiently smooth trajectories in large flight envelopes while continuously estimating the aerodynamic loads acting on the vehicle. In fact, the adaptive nature of the controller enables the use of feedback linearizing laws with limited knowledge of the aerodynamic characteristics of the vehicle. Moreover, the integrated design of position and attitude control was proven to be advantageous when aiming to derive global stability proofs for the trajectory tracking problem. It is also worth mentioning that controlling the horizontal position of the vehicle through the zero dynamics of the inner-stabilized system led to a satisfactory tracking performance which is robust to both external disturbances and model uncertainties. As future work, the authors aim to extend this control solution to the complete 6 DoF and validate it in a higher fidelity simulator, which shall account for other dynamic contributions that have been neglected so far, namely the curvature and rotation of the Earth and the MCI time derivatives.


\appendices{}              

\section{Proof of Proposition \ref{prop:rel_deg}} \label{app:proof1}       

\begin{proof}
     According to the definition in \cite{Isidori}, the nonlinear system defined by (\ref{eq:inner_xdot})-(\ref{eq:inner_y}) has a vector relative degree ${\mathbf{r_{in}}} = \{r_1,\,r_2\}$ at a point $\mathbf{x^*_{in}}$ if: \\
     \begin{enumerate}
         \item ${\mathcal{L}}_{{\mathbf{g}}_i}{\mathcal{L}}^{k}_{\mathbf{f}}{\mathbf{h}}_j({\mathbf{x}})=0$ for all $i,\,j \in  \{1,\,2\}$, $k \in \{0,\,...\,,\,r_j-2\}$, and for all $\mathbf{x_{in}}$ in a neighbourhood of $\mathbf{x^*_{in}}$;\\
         \item  The \textit{decoupling matrix} $\boldsymbol{\Lambda}_{\textbf{in}}(\mathbf{x_{in}})$ is non-singular at the point $\mathbf{x^*_{in}}$.
     \end{enumerate}
\vspace{10pt}
     
Regarding the first condition, it is straightforward to compute the respective Lie derivatives and conclude that it is verified for ${\mathbf{r_{in}}} = \{2,\,2\}$ and for all points ${\mathbf{x_{in}}}\in {\mathbb{R}}^2$. As for the second condition, the \textit{decoupling matrix} $\boldsymbol{\Lambda}_{\textbf{in}}(\mathbf{x_{in}})$ is given by
\begin{equation*}
\renewcommand{\arraystretch}{1.5}
    \boldsymbol{\Lambda}_{\textbf{in}}(\mathbf{x_{in}}) = 
    \begin{bmatrix}
        \frac{\cos{\theta}}{m} & \frac{\sin{\theta}}{m}\\
        0 & \frac{l}{j_y}
    \end{bmatrix}\,,
\end{equation*}
with determinant 
\begin{equation*}
    \text{det}\left(\boldsymbol{\Lambda}_{\textbf{in}}(\mathbf{x_{in}})\right)=\frac{\cos{\theta}\,l}{m\,j}\,.
\end{equation*}
Hence, the \textit{decoupling matrix} is invertible at any point respecting $|\theta|<\pi/2$. Thus, the inner-loop dynamics defined by (\ref{eq:inner_xdot})-(\ref{eq:inner_y}) have a well-defined vector relative degree ${\mathbf{r_{in}}}=\{2,2\}$ at any point $\mathbf{x_{in}}$ where the condition $|\theta|<\pi/2$ is satisfied.
\end{proof}

\section{Proof of Lemma \ref{eq:lemma_inner}} \label{app:lemma1}

\begin{proof}
    Let us define the following change of variables 
\begin{equation*}
    \alpha_1 = \xi_1\,,\hspace{10pt}\alpha_2 = \xi_2 - \beta\,,
\end{equation*}
and the auxiliary Lyapunov function
\begin{equation*}
    V_1(\alpha_1) = \frac{1}{2}\alpha_1^2\,.
\end{equation*}
The time derivative of the auxiliary Lyapunov function $V_1$ is
\begin{equation*}
    \dot{V}_1 = \dot{\alpha}_1\alpha_1 = \alpha_1\alpha_2 + \alpha_1\beta\,,
\end{equation*}
which for $\beta = -k_1\alpha_1$ reshapes into
\begin{equation*}
    \dot{V}_1 = -k_1\alpha_1^2 +\alpha_1\alpha_2\,.
\end{equation*}
Now let us define the Lyapunov function
\begin{equation*}
    V\left(\alpha_1,\alpha_2,\tilde{d}\right) = V_1 + \frac{1}{2}\alpha_2^2 + \frac{1}{2\gamma}{\tilde{d}}^2\,.
\end{equation*}
Its time derivative is given by
\begin{equation*}
    \dot{V} = \Dot{V}_1 + \dot{\alpha}_2\alpha_2 + \gamma^{-1}\dot{\Tilde{d}}\Tilde{d}\,,
\end{equation*}
which can be expanded to
\begin{equation*}
    \dot{V} = -k_1\alpha_1^2 +\alpha_1\alpha_2 + \left(  v + \Tilde{d} - \dot{\beta}\right)\alpha_2 + \gamma^{-1}\dot{\Tilde{d}}\Tilde{d}\,.
\end{equation*}
For $v= - \alpha_1 - k_2\alpha_2 + \dot{\beta}$ and $\dot{\Tilde{d}} = -\gamma\alpha_2$, which are equivalent to (\ref{eq:feedback_law}) and (\ref{eq:adapt_law}), the time derivative of the Lyapunov function $V$ is
\begin{equation*}
    \dot{V} = -k_1\alpha_1^2 - k_2\alpha_2^2 \leq 0\,,
\end{equation*}
which is negative semi-definite. Hence, according to Lyapunov's stability theorem \cite{Khalil}, the origin of the closed-loop system given by
\begin{equation}\label{eq:inner_edyn}
    \begin{bmatrix}
        \dot{\xi_1}\\
        \dot{\xi_2}
    \end{bmatrix}=
    \begin{bmatrix}
        0 & 1\\
        -(1+k_1k_2) & -(k_1+k_2)
    \end{bmatrix}
    \begin{bmatrix}
        \xi_1 \\
        \xi_2
    \end{bmatrix}+
    \begin{bmatrix}
        0\\
        \Tilde{d}
    \end{bmatrix}
\end{equation}
with disturbance dynamics given by (\ref{eq:adapt_law}), is globally stable. Equating the time derivative of the Lyapunov funtion $V$ to zero:
\begin{equation*}
    \dot{V}=0 \Leftrightarrow k_1\alpha_1^2 = -k_2\alpha_2^2 \Leftrightarrow \alpha_1 = \alpha_2 = 0 \Leftrightarrow \xi_1 = \xi_2 = 0\,,
\end{equation*}
since $k_1,\,k_2>0$. Thus, according to LaSalle's invariance principle \cite{Khalil}, the  origin of the error dynamics of the closed-loop system given by (\ref{eq:inner_edyn}) is globally asymptotically stable.
\end{proof}

\section{Proof of Proposition \ref{prop:out}}\label{app:prop2}

\begin{proof}
    According to the definition provided by Khalil \cite{Khalil}, the nonlinear single-input single-output (SISO) system defined by (\ref{eq:outer_xdot})-(\ref{eq:outer_y}) has relative degree $r_\text{out} = r$ at a point $\mathbf{x^{*}_\textbf{out}}$ if:\\
    \begin{enumerate}
        \item ${\mathcal{L}}_{g}{\mathcal{L}}^{i-r}_{f}h({\mathbf{x_\textbf{out}}})=0$\,, for all $i \in  \{1,\,2\,...\,,r-1\}$\,, and\\
        \item ${\mathcal{L}}_{g}{\mathcal{L}}^{r}_{f}h({\mathbf{x_\textbf{out}}})\neq0$\,,\\
    \end{enumerate}
for all $\mathbf{x_\textbf{out}}$ in a neighbourhood of $\mathbf{x^*_\textbf{out}}$. It is straightforward that both conditions are verified for $r_\text{out} = 2$ only. Likewise, the same conclusion could be reached by differentiating the output with respect to time and observe that the input only explicitly appears starting from the second derivative.
\end{proof}

\acknowledgments
The authors acknowledge the Portuguese Foundation for Science and Technology (FCT) for its financial support via the project LAETA Base Funding (DOI: 10.54499/UIDB/50022/2020) and CAPTURE (DOI:\linebreak 10.54499/PTDC/EEI-AUT/1732/2020). Pedro dos Santos holds a PhD scholarship from FCT (2023.00268.BD). The authors also gratefully acknowledge the Institute for Systems and Robotics (ISR - Lisboa) for providing access to their facilities, which were essential to the completion of this research.

\bibliographystyle{IEEEtran}

\bibliography{refs}




\thebiography
\begin{biographywithpic}
{Pedro dos Santos (M.Sc.’22)}{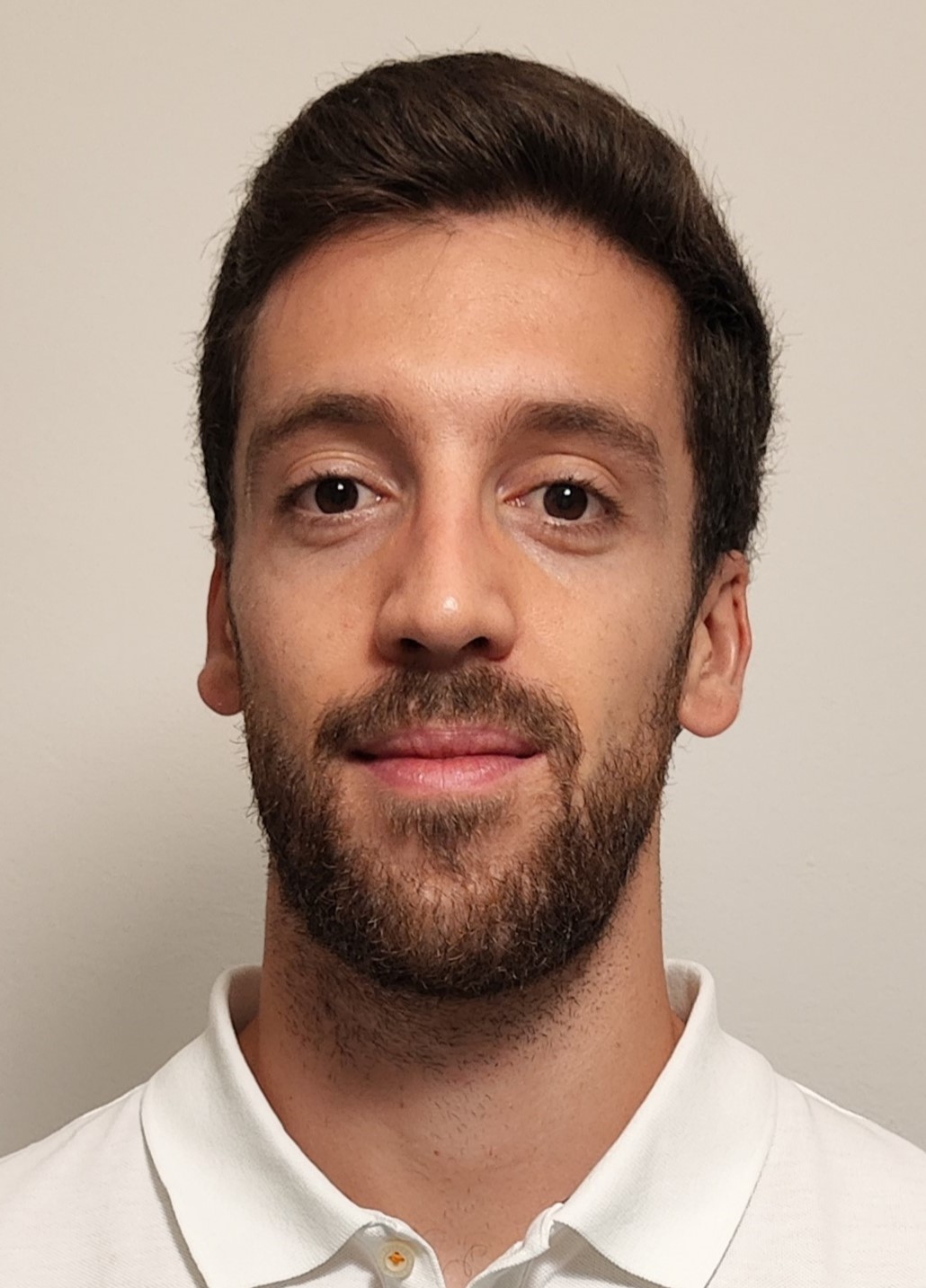}
received his B.S. and M.S. degrees in Aerospace Engineering from Instituto Superior
Técnico (IST), Lisbon, Portugal, in 2020, and 2022, respectively. He is currently pursuing a Ph.D degree in Aerospace Engineering at IST in affiliation with the Associated Laboratory
for Energy, Transports, and Aeronautics. His research interests are in the area of underactuated autonomous vehicles with a focus on non-linear and robust Guidance, Navigation, and Control (GNC) algorithms for launch vehicles. 
\end{biographywithpic} 

\begin{biographywithpic}
{Paulo Oliveira (Hab 16 Ph.D.’02 M.Sc.’91)}{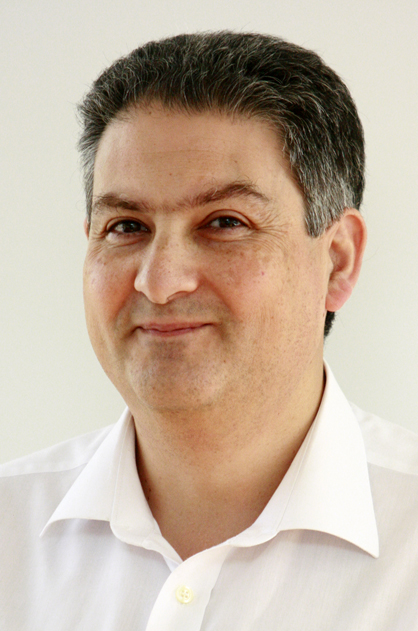}
received the “Licenciatura”,
M.S., and Ph.D. degrees in Electrical and Computer Engineering, and
the Habilitation in Mechanical Engineering from Instituto Superior
Técnico (IST), Lisbon, Portugal, in 1991, 2002, and 2016, respectively.
He is a Professor in the Department of Mechanical
Engineering of IST and Vice President for Scientific Affairs of the Associated Laboratory
for Energy, Transports, and Aeronautics. His research interests are
in the area of autonomous robotic vehicles with a focus on the fields of
estimation, sensor fusion, navigation, positioning, and mechatronics.
He is author or coauthor of more than 100 journal papers and 200 conference
communications. He participated in more than 30 European and Portuguese research projects, over the last 30 years. 

\end{biographywithpic}

\end{document}